\documentclass{article}

\usepackage{amsmath,amssymb,amsthm,graphicx,proof,color}
\usepackage[UKenglish]{babel}

\usepackage{hyperref}
\usepackage[all]{xy}

\newcommand{\bis}{\mathrel{\mathchoice%
{\raisebox{.3ex}{$\,
  \underline{\makebox[.7em]{$\leftrightarrow$}}\,$}}%
{\raisebox{.3ex}{$\,
  \underline{\makebox[.7em]{$\leftrightarrow$}}\,$}}%
{\raisebox{.2ex}{$\,
  \underline{\makebox[.5em]{\scriptsize$\leftrightarrow$}}\,$}}%
{\raisebox{.2ex}{$\,
  \underline{\makebox[.5em]{\scriptsize$\leftrightarrow$}}\,$}}}}

\newcommand{\Kbis}{\bis_{\Box}}
\newcommand{\kwbis}{\bis_{\blacktriangle}}
\newcommand{\Deltabis}{\bis_{\Delta}}

\newcommand{\equkw}{\ensuremath{\equiv_{\blacktriangle}}}
\newcommand{\equk}{\ensuremath{\equiv_{\Box}}}
\renewcommand{\triangle}{\Delta}

\hypersetup{
   colorlinks,%
   citecolor=blue,%
   filecolor=black,%
   linkcolor=red,%
   urlcolor=black
 }

\newcommand{\M}{\mathcal{M}}
\newcommand{\N}{\mathcal{N}}

\newcommand{\KwTop}{\ensuremath{\blacktriangle\top}}
\newcommand{\EquiKw}{\ensuremath{\blacktriangle\neg}}
\newcommand{\KwCon}{\ensuremath{\blacktriangle\land}}
\newcommand{\R}{\ensuremath{\texttt{R}}}
\newcommand{\TAUT}{\ensuremath{\texttt{TAUT}}}
\newcommand{\SUB}{\ensuremath{\texttt{US}}}
\newcommand{\MP}{\ensuremath{\texttt{MP}}}
\newcommand{\REKw}{\ensuremath{\texttt{RE}\blacktriangle}}
\newcommand{\KwT}{\ensuremath{\blacktriangle\texttt{T}}}
\newcommand{\KwTr}{\ensuremath{\blacktriangle4}}
\newcommand{\KwB}{\ensuremath{\blacktriangle\texttt{B}}}
\newcommand{\KwEuc}{\ensuremath{\blacktriangle5}}
\newcommand{\KwEucp}{\ensuremath{\blacktriangle5'}}

\newcommand{\Kv}{\ensuremath{\textit{Kv}}}
\newcommand{\K}{\ensuremath{\textit{K}}}
\newcommand{\BP}{\ensuremath{\textbf{P}}}

\newcommand{\CLtwo}{\ensuremath{\mathcal{L}(\Delta)}}
\newcommand{\ML}{\ensuremath{\mathcal{L}(\Box)}}
\newcommand{\CPL}{\ensuremath{\mathcal{L}}}
\newcommand{\ECPL}{\ensuremath{\mathcal{L}(\Box,\Delta,\circ,\blacktriangle)}}
\newcommand{\SNCL}{\ensuremath{\mathcal{L}(\blacktriangle)}}
\newcommand{\NCL}{\ensuremath{\mathcal{L}(\Delta)}}

\newcommand{\LEA}{\ensuremath{\mathcal{L}(\circ)}}

\newcommand{\SLCL}{\ensuremath{\mathbf{K}^\blacktriangle}}
\newcommand{\SLCLTr}{\ensuremath{\mathbf{K4}^\blacktriangle}}
\newcommand{\SLCLB}{\ensuremath{\mathbf{KB}^\blacktriangle}}
\newcommand{\SLCLEuc}{\ensuremath{\mathbb{NCL}5}}
\newcommand{\SLCLBEuc}{\ensuremath{\mathbf{KB5}^\blacktriangle}}
\newcommand{\SLCLBEucp}{\ensuremath{\mathbf{KB5'}^\blacktriangle}}
\newcommand{\lr}[1]{\langle #1 \rangle}
\newcommand{\lra}{\leftrightarrow}

\renewcommand{\phi}{\varphi}

\newtheorem{theorem}{Theorem}%��Щ��һЩpackage�ﶼ�ж���,�õ�ʱ����ϵͳ��ʾ.
\newtheorem{lemma}[theorem]{Lemma}
\newtheorem{definition}[theorem]{Definition}
\newtheorem{proposition}[theorem]{Proposition}

\newtheorem{corollary}[theorem]{Corollary}

\newtheorem{fact}[theorem]{Fact}
\newtheorem{example}[theorem]{Example}

\newcommand{\weg}[1]{}

\title{Logics of Strong Noncontingency}
\author{Jie Fan}
\date{}

\begin{document}
\maketitle

\begin{abstract}
\weg{Usually, non-contingency is interpreted as `necessary truth or necessary falsity'; otherwise, it is contingent. }Inspired by Hintikka's treatment of question embedding verbs in \cite{Hintikka:1976} and the variations of noncontingency operator, we propose a logic with {\em strong noncontingency} operator $\blacktriangle$ as the only primitive modality. A proposition is strongly noncontingent, if no matter whether it is true or false, it does it necessarily; otherwise, it is weakly contingent. This logic is not a normal modal logic, since $\blacktriangle(\phi\to\psi)\to(\blacktriangle\phi\to\blacktriangle\psi)$ is invalid. We compare the relative expressivity of this logic and other logics, such as standard modal logic, noncontingency logic, and logic of essence and accident, and investigate its frame definability. Apart from these results, we also propose a suitable notion of bisimulation for the logic of strong noncontingency, based on which we characterize this logic within modal logic and within first-order logic. We also axiomatize the logic of strong noncontingency over various frame classes. Our work is also related to the treatment of agreement operator in \cite{DBLP:journals/ndjfl/Humberstone02}.
\end{abstract}

\weg{\begin{abstract}
Usually, contingency is interpreted as `possibly true and possibly false', and non-contingency is interpreted as `necessarily true or necessarily false'. In the epistemic setting, the two notions mean, respectively, `ignorance', and `knowing whether'. Thus, an agent is ignorant about formula $\phi$, means that the agent thinks $\phi$ and $\neg\phi$ are both possible; the agent knows whether $\phi$, means that the agent knows $\phi$ or knows $\neg\phi$. \weg{However, it seems better from . }Based on this observation, we give an alternative semantics for the contingency and non-contingency operators. According to our semantics, a proposition is contingent, if no matter whether it is true or false, it could have been otherwise; otherwise, it is non-contingent, i.e., if no matter whether it is true or false, it does it necessarily. We give axiomatizations over various classes of frames.
\end{abstract}}


\section{Introduction}

In his seminal work \cite{hintikka:1962}, beyond `knowing that', Hintikka also talked about other types of knowledge, such as `knowing whether', `knowing who'. There, `John knows whether it is raining' is  interpreted as `John knows that it is raining or knows that it is not raining', and `John knows who killed Toto' is interpreted as `there is a person $b$, John knows that $b$ killed Toto' (ibid., pages 12, 132). However, Hintikka \cite{Hintikka:1976} gave a \weg{game-theoretical analysis for indirect questions}treatment of question embedding verbs. According to his interpretation, the sentence ``John knows whether it is raining'' should be equivalent to the sentence ``If it is raining, then John knows that it is raining, and if it is not raining, then John knows that it is not raining'', and the sentence ``John knows who killed Toto'' should be equivalent to the sentence ``Any person is such that if he killed Toto then John knows that he killed Toto''. These interpretations of `knowing whether' and `knowing who' are different from the ones mentioned above. His treatment also applies to many other question embedding verbs like {\em remember}. This formal analysis is criticized in \cite{Karttunen:1977}, but recently adopted by researchers, see e.g., \cite{wangetal:2010}. One motivation of this paper is to formalize Hintikka's analysis for question embedding verbs in \cite{Hintikka:1976}.

Another motivation is relevant to the variations of noncontingency operator. In a non-epistemic setting, `knowing whether' in the sense of \cite{hintikka:1962} can be read as {\em noncontingency}. A proposition is noncontingent, if it is necessarily true or necessarily false. Noncontingency operator has been studied since the 1960s, see \cite{MR66,Cresswell88,Humberstone95,DBLP:journals/ndjfl/Kuhn95,DBLP:journals/ndjfl/Zolin99,Fanetal:2014,Fanetal:2015,Fan:neighborhood}. As a variation of noncontingency operator, {\em essence} was studied in the literature on \emph{the logic of essence and accident} \cite{Marcos:2005,steinsvold:2008}.\footnote{As pointed out in \cite[page 94]{steinsvold:2008}, ``the work in \cite{Marcos:2005} can be seen as a variation and continuation of the work done on [non]contingency logics.''}  There, as the only primitive modality, ``$\phi$ is essential'' means ``if $\phi$ is true, then $\phi$ is necessarily true'', and accident operator is defined as the negation of essence.

This paper is intended to propose a strong variation of noncontingency, which we call {\em strong noncontingency}. A formula is called strongly noncontingent, if no matter whether it is true or false, it does it necessarily; otherwise, it is weakly contingent, i.e., if no matter whether it is true or false, it could have been otherwise. We denote the strong noncontingency operator by $\blacktriangle$. Intuitively, $\blacktriangle\phi$ is read as ``if $\phi$ is true, then $\phi$ is necessarily true, and if $\phi$ is false, then $\phi$ is necessarily false''. The interpretation of the new operator is in line with Hintikka's analysis for question embedding verbs in \cite{Hintikka:1976}\weg{, e.g. remember, know, tell}. For instance, in the setting of epistemic logic, `$\phi$ is strongly noncontingent' means that `if $\phi$ is true, then the agent knows that $\phi$ is true, and if $\phi$ is false, then the agent knows that $\phi$ is false'. Besides, this operator can count as a stronger version of the essence operator. Moreover, this operator is also related to an alternative semantics for agreement operator proposed in \cite{DBLP:journals/ndjfl/Humberstone02}.

The paper is organized as follows. Section \ref{sec.langandsem} defines the language and semantics of strong noncontingency logic, which is a fragment of a much larger logic. In Section \ref{sec.exp} we compare the relative expressivity of strong noncontingency logic and other logics, which turns out that the new logic is in between standard modal logic and noncontingency logic. Section \ref{sec.framedef} investigates the frame definability for strong noncontingency logic. A notion of bisimulation and that of bisimulation contraction for strong noncontingency logic are proposed in Section \ref{sec.bisandchar}. Based on the bisimulation, we characterize this logic as (strongly noncontingent) bisimulation invariant fragment of standard modal logic and of first-order logic. Sections \ref{sec.minimal} and \ref{sec.extension} axiomatize strong noncontingency logic over various classes of frames. Section \ref{sec.relatedwork} compares our work with \cite{DBLP:journals/ndjfl/Humberstone02}. We conclude with some future work in Section \ref{sec.conclu}.

\weg{Now in a general (philosophical) setting, rather than adopting the standard semantics, it is natural to interpret non-contingency operator $\blacktriangle\phi$ as ``if $\phi$ is true, then $\phi$ is necessarily true, and if $\phi$ is false, then $\phi$ is necessarily false''. At first sight, the semantics seems more complicated than the standard one, but as we will see, the axiomatizations are much simpler than those under the standard semantics. And indeed, this semantics is more intuitive than the old one.}

\section{Language and Semantics}\label{sec.langandsem}
First, we introduce various extensions of the language of classical propositional logic, although we will mainly focus on the language of strong noncontingency logic.
\weg{First, we introduce an extension of the language of classical propositional logic with noncontingency operator, which we call the language of contingency logic.}
\begin{definition}[Language $\mathcal{L}(\Box,\Delta,\circ,\blacktriangle)$] Given a set $\BP$ of propositional variables, the logical language $\mathcal{L}(\Box,\Delta,\circ,\blacktriangle)$ is defined as:
$$\phi::=\top\mid p\mid\neg\phi\mid(\phi\land\phi)\mid\Box\phi\mid\triangle\phi\mid\circ\phi\mid\blacktriangle\phi$$
Without any modal constructs, we obtain the {\em language $\CPL$ of classical propositional logic}; with the only modal construct $\Box\phi$, we obtain the {\em language \ML\ of modal logic}; with the only modal construct $\triangle\phi$, we obtain the {\em language \NCL\ of noncontingency logic};
with the only modal construct $\circ\phi$, we obtain the {\em language \LEA\ of the logic of essence and accident};
 with the only modal construct $\blacktriangle\phi$, we obtain the {\em language \SNCL\ of strong noncontingency logic}. Given any language $\mathcal{L}$ in question, if $\phi\in\mathcal{L}$, we say $\phi$ is an $\mathcal{L}$-formula.
\end{definition}

We always omit the parentheses from formulas whenever no confusion arises. Formulas $\Box\phi,\triangle\phi,\circ\phi,\blacktriangle\phi$ express, respectively, `it is necessary that $\phi$', `it is noncontingent that $\phi$', `it is essential that $\phi$', `it is {\em strongly noncontingent} that $\phi$'. Other operators are defined as usual; in particular, $\Diamond\phi,\nabla\phi,\bullet\phi,\blacktriangledown\phi$ are defined as, respectively, $\neg\Box\neg\phi,\neg\triangle\phi,\neg\circ\phi,\neg\blacktriangle\phi$.

\begin{definition}[Model]
A {\em model} $\M$ is a triple $\lr{S,R,V}$, where $S$ is a nonempty set of possible worlds, $R$ is a binary relation over $S$, $V$ is a valuation function from $\BP$ to $\mathcal{P}(S)$. A {\em frame} $\mathcal{F}$ is a pair $\lr{S,R}$, i.e. a model without a valuation. Given $s\in S$, $(\M,s)$ is  a {\em pointed model}, and $(\mathcal{F},s)$ is a {\em pointed frame}. We omit the parentheses around $(\M,s)$ and $(\mathcal{F},s)$ whenever convenient. Sometimes we write $s\in\M$ for $s\in S$. And we denote by $R(s)$ the set of successors of $s$, that is, $R(s)=\{t\in S\mid sRt\}$. Model $\M$ is said to be a $\mathcal{K}$-model (resp. $\mathcal{D}$-model, $\mathcal{T}$-model, $\mathcal{B}$-model, $4$-model, $5$-model, $\mathcal{TB}$-model, $\mathcal{S}4$-model, $\mathcal{S}5$-model) if $R$ is arbitrary (resp. serial, reflexive, symmetric, transitive, Euclidean, reflexive and symmetric, reflexive and transitive, reflexive and Euclidean). A $\mathcal{K}$-frame and the like are defined similarly. %$\M$ is said to be a $\mathcal{TB}$-model (resp. $\mathcal{S}4$-model, $\mathcal{S5}$-model) if $\M$ is
\end{definition}

\begin{definition}[Semantics] Given a model $\mathcal{M}=\langle S,R,V\rangle$ and $s\in S$, the semantics of \ECPL\ is defined recursively as:\footnote{We here use the notation $\&,\forall,\Rightarrow,\Leftrightarrow$, respectively, to stand for the metalanguage `and', `for all', `if $\cdots$ then $\cdots$', `if and only if'.}
\[\begin{array}{|lcl|}
\hline
\mathcal{M},s\vDash\top&\Leftrightarrow&\text{true}\\
\mathcal{M},s\vDash p & \Leftrightarrow & s\in V(p)\\
\mathcal{M},s\vDash\neg\phi & \Leftrightarrow & \mathcal{M},s\nvDash\phi\\
\mathcal{M},s\vDash\phi\land\psi & \Leftrightarrow & \mathcal{M},s\vDash\phi ~\&~ \mathcal{M},s\vDash\psi\\
\mathcal{M},s\vDash\Box\phi& \Leftrightarrow & \forall t(sRt\Rightarrow\mathcal{M}, t\vDash\phi)\\
\mathcal{M},s\vDash\triangle\phi&\Leftrightarrow& \forall t_1,t_2((sRt_1~\&~sRt_2)\Rightarrow(\mathcal{M},t_1\vDash\phi\Leftrightarrow \mathcal{M},t_2\vDash\phi))\\
\mathcal{M},s\vDash\circ\phi&\Leftrightarrow&
(\mathcal{M},s\vDash\phi\Rightarrow\forall t(sRt\Rightarrow\mathcal{M}, t\vDash\phi))\\
\mathcal{M},s\vDash\blacktriangle\phi&\Leftrightarrow&(\mathcal{M},s\vDash\phi\Rightarrow \forall t(sRt\Rightarrow \mathcal{M},t\vDash\phi)) ~\& \\
&& (\mathcal{M},s\nvDash\phi\Rightarrow\forall t(sRt\Rightarrow\mathcal{M}, t\nvDash\phi))\\
\hline
\end{array}\]
If $\M,s\vDash\phi$, we say $\phi$ is {\em true} at $s$, and sometimes write $s\vDash\phi$ if $\M$ is clear; if for all $s\in\M$ we have $\M,s\vDash\phi$, we say $\phi$ is {\em valid on $\M$} and write $\M\vDash\phi$; if for all $\M$ based on $\mathcal{F}$ we have $\M\vDash\phi$, we say $\phi$ is {\em valid on $\mathcal{F}$} and write $\mathcal{F}\vDash\phi$; if for all $\mathcal{F}$ we have $\mathcal{F}\vDash\phi$, we say $\phi$ is {\em valid} and write $\vDash\phi$. Given $\Gamma\subseteq\ECPL$, if for all $\phi\in\Gamma$ we have $\M,s\vDash\phi$, we say $\M,s\vDash\Gamma$, and similarly for model validity/frame validity/validity. If there exists $(\M,s)$ such that $\M,s\vDash\phi$, we say $\phi$ is {\em satisfiable}. If $\Gamma\cup\{\neg\phi\}$ is unsatisfiable in $\mathcal{F}$, we say $\Gamma$ {\em entails} $\phi$ over $\mathcal{F}$ and write $\Gamma\vDash_\mathcal{F}\phi$. Given any two pointed models $(\M,s)$ and $(\N,t)$, if they satisfy the same $\ML$-formulas, we say they are {\em $\Box$-equivalent}, notation: $(\M,s)\equk(\N,t)$; if they satisfy the same $\SNCL$-formulas, we say they are {\em $\blacktriangle$-equivalent}, notation: $(\M,s)\equkw(\N,t)$.
\end{definition}

It is not hard to show the following validities.
\begin{fact}\label{fact.one}
$\vDash\blacktriangle\phi\leftrightarrow(\phi\to\Box\phi)\land(\neg\phi\to\Box\neg\phi)$, $\vDash\blacktriangledown\phi\leftrightarrow(\phi\land\Diamond\neg\phi)\lor(\neg\phi\land\Diamond\phi)$, $\vDash\blacktriangle\phi\leftrightarrow\circ\phi\land\circ\neg\phi$, $\vDash\blacktriangle\phi\leftrightarrow\blacktriangle\neg\phi$, $\vDash\blacktriangledown\phi\leftrightarrow\blacktriangledown\neg\phi$, $\vDash\triangle\phi\leftrightarrow\Box\phi\lor\Box\neg\phi$.
\end{fact}

\weg{The following proposition says that the alternative semantics of non-contingency operator is strictly stronger than its standard semantics.}
The following proposition says that
our new operator $\blacktriangle$ is indeed strictly
stronger than the standard non-contingency operator, that is why we call the language $\SNCL$ the language of strong noncontingency logic. And also, the new operator is strictly stronger than the essence operator.
\begin{proposition}\label{prop.stronger}
$\vDash\blacktriangle\phi\to\triangle\phi$, $\vDash\blacktriangle\phi\to\circ\phi$, but $\nvDash\triangle\phi\to\blacktriangle\phi$ and $\nvDash\circ\phi\to\blacktriangle\phi$.
\end{proposition}

\begin{proof}
The validity of $\blacktriangle\phi\to\circ\phi$ is immediate from the semantics.
Let $\mathcal{M}=\langle S,R,V\rangle$ be a model and $s\in S$. Assume that $\mathcal{M},s\vDash\blacktriangle\phi$, then by Fact~\ref{fact.one}, we have $\mathcal{M},s\vDash(\phi\to\Box\phi)\land(\neg\phi\to\Box\neg\phi)$. Since $s\vDash\phi\vee\neg\phi$, we can show that $s\vDash\Box\phi\vee\Box\neg\phi$, thus it is immediate from Fact~\ref{fact.one} that $s\vDash\triangle\phi$. Therefore $\vDash\blacktriangle\phi\to\triangle\phi$.

For the invalidity, consider the following model $\mathcal{N}$:
$$
\xymatrix{s:p\ar[rr]&&t:\neg p}
$$
Since $s$ has only one successor, we have $\mathcal{N},s\vDash\triangle p$. However, one can check $s\nvDash p\to\Box p$, using Fact \ref{fact.one} we get $s\nvDash\blacktriangle p$, thus
$\nvDash\triangle p\to\blacktriangle p$.
Besides, $s\vDash\circ\neg p$, but $s\nvDash\blacktriangle\neg p$, thus $\nvDash\circ\neg p\to\blacktriangle\neg p$.
\end{proof}

As a corollary, we have
\begin{corollary}
$\vDash\nabla\phi\to\blacktriangledown
\phi$, $\vDash\bullet\phi\to\blacktriangledown\phi$, but $\nvDash\blacktriangledown\phi
\to\nabla\phi$ and $\nvDash\blacktriangledown\phi\to\bullet\phi$.
\end{corollary}

The following proposition states that necessity and strong noncontingency are \emph{almost} equivalent. We thus call the formula $\phi\to(\Box\phi\leftrightarrow\blacktriangle\phi)$ `almost-equivalent' schema (AE).
\begin{proposition}\label{prop.ad}
$\vDash\phi\to(\Box\phi\leftrightarrow\blacktriangle\phi)$.
\end{proposition}

\begin{proof}
Let $\mathcal{M}=\langle S,R,V\rangle$ be a model, and assume that $\mathcal{M},s\vDash\phi$. From the assumption and Fact~\ref{fact.one} it is clear that $\vDash\blacktriangle\phi\to\Box\phi$. We only need to show $\mathcal{M},s\vDash\Box\phi\to\blacktriangle\phi$. For this, suppose that $\mathcal{M},s\vDash\Box\phi$, then obviously $\mathcal{M},s\vDash\phi\to\Box\phi$; moreover, from the assumption it follows that $\mathcal{M},s\vDash\neg\phi\to\Box\neg\phi$, thus $\mathcal{M},s\vDash(\phi\to\Box\phi)\land(\neg\phi\to\Box\neg\phi)$. According to Fact~\ref{fact.one} again, we conclude that $\mathcal{M},s\vDash\blacktriangle\phi$, as desired.
\end{proof}

The schema AE motivates us to propose the desired canonical relation for the axiomatic systems below.

The logic $\SNCL$ is {\em not} normal, because $\blacktriangle(\phi\to\psi)\to(\blacktriangle\phi\to\blacktriangle\psi)$ is invalid. Neither is $\SNCL$ monotonic, since $\vDash\phi\to\psi$ does not imply $\vDash\blacktriangle\phi\to\blacktriangle\psi$, as illustrated below.
\[
\xymatrix{\M:\ \ \ s:\neg p, q\ar[rr]&& t:\neg p,\neg q}\\
\]
It is not hard to show that $\M,s\vDash\blacktriangle(p\to q)\land\blacktriangle p$~but~$\M,s\nvDash\blacktriangle q$. Moreover, $\vDash p\land q\to q$ but $\nvDash\blacktriangle (p\land q)\to\blacktriangle q$.

Although $\blacktriangle(\phi\to\psi)\to(\blacktriangle\phi\to\blacktriangle\psi)$ is invalid, we have a weaker validity.

\begin{proposition}
$\blacktriangle \phi\land\blacktriangle(\phi\to \psi)\land \phi\to\blacktriangle \psi$ is valid.
\end{proposition}

\begin{proof}
Given an arbitrary model $\mathcal{M}=\langle S,R,V\rangle$ and any $s\in S$, suppose that $\mathcal{M},s\vDash\blacktriangle \phi\land\blacktriangle(\phi\to \psi)\land \phi$, to show $\mathcal{M},s\vDash\blacktriangle \psi$. By supposition and Prop.~\ref{prop.ad}, we have $s\vDash\Box\phi$. If $s\vDash\psi$, then $s\vDash\phi\to\psi$, from which and supposition $s\vDash\blacktriangle(\phi\to \psi)$ and Prop.~\ref{prop.ad}, it follows that $s\vDash\Box(\phi\to\psi)$, thus we can get $s\vDash\Box\psi$. If $s\vDash\neg\psi$, then $s\vDash\neg(\phi\to\psi)$, from which and supposition $s\vDash\blacktriangle(\phi\to \psi)$ and Fact~\ref{fact.one}, it follows that $s\vDash\Box\neg(\phi\to\psi)$, thus we can get $s\vDash\Box\neg\psi$. We have thus shown that $s\vDash(\psi\to\Box\psi)\land(\neg\psi\to\Box\neg\psi)$, then by Fact~\ref{fact.one} again, we conclude that $s\vDash\blacktriangle\psi$.
\end{proof}

\section{Relative expressivity}\label{sec.exp}

In this section, we compare the relative expressivity among different logics $\SNCL$, $\ML$ and $\CLtwo$\weg{of different logics, including $\SNCL$ and $\ML$, $\SNCL$ and $\CLtwo$, and $\ML$ and $\CLtwo$}.  A related technical definition is introduced as follows.

%First, let us give a technical definition.
\begin{definition}[Expressivity]
Given logical languages $\mathcal{L}_1$ and $\mathcal{L}_2$ that are interpreted on the same class $\mathbb{M}$ of models,
\begin{itemize}
\item $\mathcal{L}_2$ is {\em at least as expressive as} $\mathcal{L}_1$, notation: $\mathcal{L}_1\preceq \mathcal{L}_2$, if for any $\phi\in \mathcal{L}_1$, there exists $\psi\in\mathcal{L}_2$ such that for all $(\M,s)\in\mathbb{M}$, we have $\M,s\vDash\phi\leftrightarrow\psi$.
\item $\mathcal{L}_1$ and $\mathcal{L}_2$ are {\em equally expressive}, notation: $\mathcal{L}_1\equiv\mathcal{L}_2$, if $\mathcal{L}_1\preceq \mathcal{L}_2$ and $\mathcal{L}_2\preceq\mathcal{L}_1$.
\item $\mathcal{L}_1$ is {\em less expressive than} $\mathcal{L}_2$, or $\mathcal{L}_2$ is {\em more expressive than} $\mathcal{L}_1$, notation: $\mathcal{L}_1\prec\mathcal{L}_2$, if $\mathcal{L}_1\preceq \mathcal{L}_2$ and $\mathcal{L}_2\not\preceq\mathcal{L}_1$.
\end{itemize}
\end{definition}

\subsection{$\SNCL$ vs. $\ML$}
We first compare the relative expressivity of $\SNCL$ and $\ML$.
\begin{proposition}\label{prop.lessexp}
\SNCL\ is less expressive than \ML\ on the class of $\mathcal{K}$-models, $\mathcal{B}$-models, $4$-models, $5$-models.
\end{proposition}

\begin{proof}
Define a translation $t$ from $\SNCL$ to \ML:
\[
\begin{array}{lll}
t(\top)&=&\top\\
t(p)&=&p\\
t(\neg\phi)&=&\neg t(\phi)\\
t(\phi\land\psi)&=&t(\phi)\land t(\psi)\\
t(\blacktriangle\phi)&=&(t(\phi)\to\Box t(\phi))\land(\neg t(\phi)\to\Box\neg t(\phi))\\
\end{array}
\]

It is clear from Fact~\ref{fact.one} that $t$ is a truth-preserving translation. Therefore \ML\ is at least as expressive as \SNCL.

Now consider the following pointed models $(\mathcal{M},s)$ and $(\mathcal{N},t)$, which can be distinguished by an \ML-formula $\Box\bot$, but cannot be distinguished by any $\SNCL$-formulas:

\medskip

$$
\xymatrix{\mathcal{M}:\ \ \ s:p\ar@(ur,ul) & & & \mathcal{N}:\ \ \ t:p}
$$

It is easy to check $\mathcal{M}$ and $\mathcal{N}$ are both symmetric, transitive, and Euclidean. By induction we prove that for all $\phi\in\SNCL$, $\mathcal{M},s\vDash\phi$ iff $\mathcal{N},t\vDash\phi$. The base cases and boolean cases are straightforward. For the case of $\blacktriangle\phi$, it is not hard to show that $\mathcal{M},s\vDash\blacktriangle\phi$ and $\mathcal{N},t\vDash\blacktriangle\phi$ (note that here we do not need to use the induction hypothesis), thus $\mathcal{M},s\vDash\blacktriangle\phi$ iff $\mathcal{N},t\vDash\blacktriangle\phi$, as desired.
\end{proof}

As for the case of $\mathcal{D}$-models, the result about the relative expressivity of $\SNCL$ and $\ML$ is same as previous, but the proof is much more sophisticated, which needs {\em simultaneous induction}.

\begin{proposition}\label{prop.lessexp-d}
\SNCL\ is less expressive than $\ML$ on the class of $\mathcal{D}$-models.
\end{proposition}

\begin{proof}
By the translation $t$ in the proof of Proposition \ref{prop.lessexp}, we have $\SNCL\preceq\ML$.

Consider the following pointed models $(\M,s)$ and $(\N,s')$, which can be distinguished by an $\ML$-formula $\Box\Box p$, but cannot be distinguished by any $\SNCL$-formulas:
$$
\xymatrix{\mathcal{M}:\ \ \ s:p\ar[rr]&&t:\neg p\ar@(ur,ul)\ar[ll] & & \mathcal{N}:\ \ \ s':p\ar[rr]&& t':\neg p\ar[ll]}
$$

It is not hard to see that $\M$ and $\N$ are both serial. By induction on $\phi\in\SNCL$, we show simultaneously that for all $\phi$, (i) $\M,s\vDash\phi$ iff $\N,s'\vDash\phi$, and (ii) $\M,t\vDash\phi$ iff $\N,t'\vDash\phi$. The nontrivial case is $\blacktriangle\phi$.

For (i), we have the following equivalences:
\[
\begin{array}{lll}
\M,s\vDash\blacktriangle\phi&\stackrel{\text{semantics}}\Longleftrightarrow& s\vDash\phi\text{ iff }t\vDash\phi\\
&\stackrel{\text{IH for (i)}}\Longleftrightarrow& s'\vDash\phi\text{ iff }t\vDash\phi\\
&\stackrel{\text{(ii)}}\Longleftrightarrow&s'\vDash\phi\text{ iff }t'\vDash\phi\\
&\stackrel{\text{semantics}}\Longleftrightarrow& \N,s'\vDash\blacktriangle\phi\\
\\
\M,t\vDash\blacktriangle\phi&\stackrel{\text{semantics}}\Longleftrightarrow& s\vDash\phi\text{ iff }t\vDash\phi\\
&\stackrel{\text{IH for (ii)}}\Longleftrightarrow& s\vDash\phi\text{ iff }t'\vDash\phi\\
&\stackrel{\text{(i)}}\Longleftrightarrow&s'\vDash\phi\text{ iff }t'\vDash\phi\\
&\stackrel{\text{semantics}}\Longleftrightarrow& \N,t'\vDash\blacktriangle\phi\\
\end{array}\]

Therefore, $(\M,s)$ and $(\N,s')$ cannot be distinguished by any $\SNCL$-formulas.
\end{proof}

However, on the $\mathcal{T}$-models, the situation is different.
\begin{proposition}\label{prop.equallyexp}
\SNCL\ and \ML\ are equally expressive on the class of $\mathcal{T}$-models.
\end{proposition}

\begin{proof}
By the translation $t$ in the proof of Prop.~\ref{prop.lessexp}, we have $\SNCL\preceq\ML$. Besides, define another translation $t'$ from \ML\ to \SNCL, where the base cases and Boolean cases are similar to the corresponding cases for $t$, and $t'(\Box\phi)=\blacktriangle t'(\phi) \land t'(\phi)$. It is straightforward to show that $t'$ is a truth-preserving translation, due to the validity of $\Box\phi\to\phi$ and Prop.~\ref{prop.ad}. Thus $\ML\preceq\SNCL$, and therefore $\SNCL\equiv\ML$.
\end{proof}

\subsection{$\SNCL$ vs. $\CLtwo$}

We can also compare the relative expressivity of \SNCL\ and \CLtwo.
\begin{proposition}\label{prop.exp.clncl}
\CLtwo\ is less expressive than \SNCL\ on the class of $\mathcal{K}$-models, $\mathcal{D}$-models, $\mathcal{B}$-models, $4$-models, $5$-models.
\end{proposition}

\begin{proof}

For the cases of $\mathcal{K},\mathcal{D},4,5$, consider the following models:
\medskip
$$
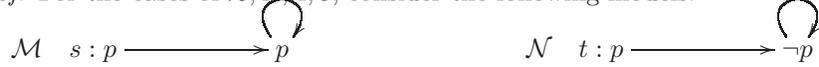
\xymatrix{\mathcal{M}\ \ \ s:p\ar[rr]& & p\ar@(ul,ur)&&& \mathcal{N}\ \ \ t:p\ar[rr]&& \neg p\ar@(ul,ur)}
$$
It is not hard to show that $\mathcal{M},s\vDash\blacktriangle p$ but $\mathcal{N},t\nvDash\blacktriangle p$. Therefore $\blacktriangle p$ can distinguish the two models.

However, note that both of $s$ and $s'$ satisfy the same proposition variables, and have only one successor. Then we can show by induction that, for any $\phi\in\CLtwo$, $\mathcal{M},s\vDash\phi$ iff $\mathcal{N},t\vDash\phi$. This means that the two models cannot be distinguished by any formulas of \CLtwo.

\medskip

For the case of $\mathcal{B}$-models, consider the following models:
\medskip
$$
\xymatrix{\mathcal{M'}\ \ \ s':p\ar[rr]& & p\ar[ll]&&& \mathcal{N'}\ \ \ t':p\ar[rr]&& \neg p\ar[ll]}
$$
Similarly, we can show that $\mathcal{M'},s'\vDash\blacktriangle p$ but $\mathcal{N'},t'\nvDash\blacktriangle p$; however, $\mathcal{M'},s'\vDash\phi$ iff $\mathcal{N'},t'\vDash\phi$ for any $\phi\in\CLtwo$.
\end{proof}

\weg{It is shown in \cite{DBLP:journals/sLogica/Demri97} that on the class of $\mathcal{T}$-models, \ML\ and \CLtwo\ are equally expressive. By Prop.~\ref{prop.equallyexp}, it is immediate that}
\begin{proposition}\label{prop.exp.t}
\SNCL\ and \CLtwo\ are equally expressive on the class of $\mathcal{T}$-models.
\end{proposition}

\begin{proof}
We only need to show that, on the class of $\mathcal{T}$-models, $\vDash\blacktriangle\phi\leftrightarrow\Delta\phi$. The validity of $\blacktriangle\phi\to\Delta\phi$ was shown in Prop.~\ref{prop.stronger}. For the validity of the other direction, given any reflexive model $\M=\lr{S,R,V}$ and any $s\in S$, suppose towards contradiction that $\M,s\vDash\Delta\phi$ but $s\nvDash\blacktriangle\phi$. Then from the second supposition, there exists $t$ such that $sRt$ and $(s\vDash\phi\not\Leftrightarrow t\vDash\phi)$. Since $R$ is reflexive, we have $sRs$. We have thus found two successors of $s$ which disagree on the value of $\phi$, hence $s\nvDash\Delta\phi$, contrary to the first supposition.
\end{proof}

From Propositions \ref{prop.lessexp}-\ref{prop.exp.t}, we obtain the results on the relative expressivity of $\CLtwo$ and $\ML$, which was shown in \cite[Section 3.1]{Fanetal:2015}.
\begin{corollary}
\CLtwo\ is less expressive than \ML\ on the class of $\mathcal{K}$-models, $\mathcal{D}$-models, $4$-models, $\mathcal{B}$-models, $5$-models, but they are equally expressive on the class of $\mathcal{T}$-models.
\end{corollary}

\weg{\subsection{$\SNCL$ vs. $\LEA$}}

\section{Frame definability}\label{sec.framedef}
As have been shown, all of the five basic frame properties can be captured by standard modal logic, e.g., the property of reflexivity can be captured by the $\ML$-formula $\Box p\to p$, refer to e.g. \cite{blackburnetal:2001}. It is also shown in \cite[Corollary 4.4]{DBLP:journals/ndjfl/Zolin99} that these five basic frame properties cannot be defined by $\CLtwo$-formulas. In this section, we do the same job for the logic $\SNCL$. It turns out that the results are inbetween: some of the five basic frame properties are undefinable in $\SNCL$, some is definable.

\begin{definition}[Frame definability] Let $\mathbb{F}$ be a class of frames. We say that $\mathbb{F}$ is {\em definable} in $\SNCL$, if there exists $\Gamma\subseteq \SNCL$ that defines it, viz., for any $\mathcal{F}$, $\mathcal{F}$ is in $\mathbb{F}$ iff $\mathcal{F}\vDash\Gamma$. In this case we also say $\Gamma$ defines the property of $\mathbb{F}$. If $\Gamma$ is a singleton set (e.g. $\{\phi\}$), then we write $\mathcal{F}\vDash\phi$ for $\mathcal{F}\vDash\{\phi\}$.
\end{definition}

Let $\mathcal{F}=\langle S,R\rangle$. Say $R$ is a {\em coreflexive relation}\footnote{The terminology {\em coreflexive relation} is taken from the website \url{http://en.wikipedia.org/wiki/Coreflexive_relation}.}, if for all $s,t\in S$, $sRt$ implies $s=t$; intuitively, every point in $S$ can at most `see' itself. Say $\mathcal{F}$ is a coreflexive frame, if $R$ is coreflexive.

\begin{proposition}\label{prop.coref}
Let $\mathcal{F},\mathcal{F}'$ be both coreflexive frames. Then given any $\phi\in\SNCL$, we have $\mathcal{F}\vDash\phi$ iff $\mathcal{F}'\vDash\phi$.
\end{proposition}

\begin{proof}
Let $\mathcal{F}=\langle S,R\rangle,\mathcal{F}'=\langle S',R'\rangle$ be both coreflexive frames, and let $\phi\in\SNCL$.

Suppose that $\mathcal{F}\nvDash\phi$. Then there exists $\M=\langle \mathcal{F},V\rangle$ and $s\in S$ such that $\M,s\nvDash\phi$. Because $S'$ is nonempty, we may assume that $s'\in S'$. Define a valuation $V'$ on $\mathcal{F}'$ as $p\in V'(s')$ iff $p\in V(s)$ for all $p\in\BP$. Since $\mathcal{F}$ and $\mathcal{F}'$ are both coreflexive, both $s$ and $s'$ can at most see itself. By induction on $\psi\in\SNCL$, we can show that for every $\psi\in\SNCL$, $\M,s\vDash\psi$ iff $\M',s'\vDash\psi$ (note that for the case of $\blacktriangle\phi$, we do not need the induction hypothesis here). Thus $\M',s'\nvDash\phi$, and therefore $\mathcal{F}'\nvDash\phi$. The converse is similar.
\end{proof}

\begin{proposition}
The frame properties of seriality, reflexivity, and endpointedness are not definable in $\SNCL$.
\end{proposition}

\begin{proof}
We adjust the example in Prop.~\ref{prop.lessexp}, i.e., consider the following two frames:

$$
\xymatrix{\mathcal{F}:\hspace{-.2cm} & s\ar@(ur,ul) & & & \mathcal{F}':\hspace{-.2cm} & t}
$$

Both frames are coreflexive. By Prop.~\ref{prop.coref}, we have: for every $\phi\in\SNCL$, $\mathcal{F}\vDash\phi$ iff $\mathcal{F}'\vDash\phi$. Notice that $\mathcal{F}$ is serial (resp. reflexive), while $\mathcal{F}'$ is not; $\mathcal{F}'$ is an endpoint frame, but $\mathcal{F}$ is not.

The argument now continues as follows. Consider seriality: if seriality were defined by a set $\Gamma$ of $\SNCL$-formulas, then as $\mathcal{F}$ is serial, then $\mathcal{F}\vDash\Gamma$. As $\mathcal{F}$ and $\mathcal{F}'$ satisfy the same $\SNCL$-formulas, we should also have $\mathcal{F}'\vDash\Gamma$, thus $\mathcal{F}'$ should be serial. However, $\mathcal{F}'$ is not serial. Therefore, seriality is not definable in $\SNCL$.

The similar argument goes for the other cases.
\end{proof}

Transitivity and Euclidicity are also undefinable with $\SNCL$, though the proof is a little more tricky.
\begin{proposition}
The frame properties of transitivity and Euclidicity are not definable in $\SNCL$.
\end{proposition}

\begin{proof}
Consider the following frames:
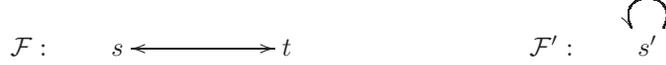
\[
\xymatrix{\mathcal{F}:\hspace{-.2cm} & s\ar[rr]&&t\ar[ll] & & & \mathcal{F}':\hspace{-.2cm} & s'\ar@(ur,ul)}
\]

We first show that ($\star$): for any $\phi\in\SNCL$, $\mathcal{F},s\vDash\phi$ iff $\mathcal{F},t\vDash\phi$. The proof is by induction on $\phi$. The nontrivial case is $\blacktriangle\phi$. In this case $\mathcal{F},s\vDash\blacktriangle\phi$ iff for all $\M$ based on $\mathcal{F}$, ($\M,s\vDash\phi$ iff $\M,t\vDash\phi$) iff $\mathcal{F},t\vDash\blacktriangle\phi$.

We have thus shown ($\star$), which implies $\mathcal{F}\vDash\blacktriangle\phi$. Besides, $\mathcal{F}'\vDash\blacktriangle\phi$. We can now show by induction on $\phi\in\SNCL$ that $\mathcal{F}\vDash\phi$ iff $\mathcal{F}'\vDash\phi$ (note that for the case of $\blacktriangle\phi$, we do not need the induction hypothesis here).

Observe that $\mathcal{F}'$ is transitive and Euclidean, while $\mathcal{F}$ is not ($sRt,tRs$ but it is not the case that $sRs$, thus $\mathcal{F}$ is not transitive; $sRt,sRt$ but it is not the case that $tRt$, thus $\mathcal{F}$ is not Euclidian).

The argument now continues as follows. Consider transitivity: if transitivity were defined by a set $\Gamma$ of $\SNCL$-formulas, then as $\mathcal{F}'$ is transitive, then $\mathcal{F}'\vDash\Gamma$. As $\mathcal{F}$ and $\mathcal{F}'$ satisfy the same $\SNCL$-formulas, we should also have $\mathcal{F}\vDash\Gamma$, thus $\mathcal{F}$ should be transitive. However, $\mathcal{F}$ is not transitive. Therefore, transitivity is not definable in $\SNCL$.

The argument is similar for the case of Euclidicity.
\end{proof}

\weg{\begin{proposition}
The frame property of transitivity is definable in $\SNCL$.
\end{proposition}

\begin{proof}
We claim that transitivity is defined by Axiom $\KwTr$, i.e., $$\mathcal{F}\vDash\forall x\forall y\forall z(xRy\land yRz\to xRz)\text{ iff }\mathcal{F}\vDash\blacktriangle p\to\blacktriangle\blacktriangle p.$$
By Proposition \ref{prop.validities}, we need only show the direction from right to left.

Suppose that $\mathcal{F}$ is not transitive, that is, there exist $s,t,u$ such that $sRt,tRu$ but not $sRu$. Define $V$ as a valuation on $\mathcal{F}$ such that $V(p)=\{w\mid sRw\}\cup\{s\}$. Now it is not hard to show that $s\vDash\blacktriangle p$. On the other hand, since $u\not$
\end{proof}}

In spite of so many undefinability results, we have the following definability results.
\begin{proposition}\label{prop.definable-b}
The property of symmetry is definable in $\SNCL$.
\end{proposition}

\begin{proof}
Let $\mathcal{F}=\langle S,R\rangle$ be a frame. We claim that symmetry is defined by the $\SNCL$-formula $p\to\blacktriangle(\blacktriangle p\to p)$, i.e., \weg{that symmetry is definable by Axiom $\KwB$, i.e.,} $$\mathcal{F}\vDash\forall x\forall y(xRy\to yRx)\text{ iff }\mathcal{F}\vDash p\to\blacktriangle(\blacktriangle p\to p).$$
%By Proposition \ref{prop.validities}, we need only show the direction from right to left.

Suppose that $\mathcal{F}$ is symmetric. Given any model $\mathcal{M}$ based on $\mathcal{F}$ and any $s\in S$, assume towards contradiction that $\mathcal{M},s\vDash p$ but $\mathcal{M},s\nvDash\blacktriangle(\blacktriangle p\to p)$. Then $s\vDash\blacktriangle p\to p$, and thus there exists $t\in S$ such that $sRt$ and $\mathcal{M},t\vDash\neg(\blacktriangle p\to p)$, hence $t\vDash\blacktriangle p$ and $t\vDash\neg p$, i.e., $t\vDash \neg p\land\blacktriangle \neg p$. Since $sRt$ and $R$ is symmetric, we have $tRs$. By semantics of $\blacktriangle$, we conclude that $s\vDash \neg p$, contrary to the assumption.

Suppose that $\mathcal{F}$ is not symmetric, that is, there are $s,t$ such that $sRt$ but not $tRs$. Clearly, $s\neq t$. Define a valuation $V$ on $\mathcal{F}$ as $V(p)=\{s\}$. Then $\lr{\mathcal{F},V},s\vDash p$, thus $s\vDash \blacktriangle p\to p$. Since $t\neq s$, and $s$ is not reachable from $t$, it follows that $t\nvDash p$, and for all $u$ such that $tRu$, we have $u\neq s$, thus $u\nvDash p$. This implies $t\vDash \blacktriangle p$, and then $t\nvDash \blacktriangle p\to p$, and hence $s\nvDash\blacktriangle (\blacktriangle p\to p)$. Therefore $\lr{\mathcal{F},V},s\nvDash p\to\blacktriangle(\blacktriangle p\to p)$, which implies that $\mathcal{F}\nvDash p\to\blacktriangle(\blacktriangle p\to p)$.
\end{proof}

\begin{proposition}
The frame property of coreflexivity is definable in $\SNCL$.
\end{proposition}

\begin{proof}
Consider the $\SNCL$-formula $\blacktriangle p$. We show that for each frame $\mathcal{F}=\langle S,R\rangle$, $\mathcal{F}\vDash\blacktriangle p$ iff $\mathcal{F}\vDash\forall x\forall y(xRy\to x=y)$.

`If': Suppose that $\mathcal{F}\vDash\forall x\forall y(xRy\to x=y)$. Then given any $\M=\langle \mathcal{F},V\rangle$ and any $s\in S$, if for each $t$ with $sRt$, we have $s=t$, then $s\vDash p$ iff $t\vDash p$, and thus $s\vDash \blacktriangle p$, therefore $\mathcal{F}\vDash\blacktriangle p$.

`Only if': Suppose that $\mathcal{F}\nvDash\forall x\forall y(xRy\to x=y)$. Then there exist $s,t\in S$ such that $sRt$ but $s\neq t$. Define a valuation $V$ on $\mathcal{F}$ such that $V(p)=\{s\}$, then $s\vDash p$ but $t\nvDash p$, and thus $\langle\mathcal{F},V\rangle,s\nvDash\blacktriangle p$, therefore $\mathcal{F}\nvDash\blacktriangle p$.
\end{proof}

As mentioned before, all of the five basic frame properties are definable in $\ML$, but not definable in $\CLtwo$; however, in $\SNCL$, some of the properties in question are undefinable, some of them is definable. This can be explain in terms of the fact that the expressivity of $\SNCL$ is in between that of $\ML$ and that of $\CLtwo$ (see Prop.~\ref{prop.lessexp} and Prop.~\ref{prop.exp.clncl}).

\section{Bisimulation and Characterization results}\label{sec.bisandchar}

In this section, we propose a notion of bisimulation for the logic $\SNCL$, and define a notion of bisimulation contraction for the same logic. Based on this bisimulation, we characterize the strong noncontingency logic within standard modal logic and within first-order logic.

\subsection{Bisimulation}
First, let us introduce the standard notion of bisimulation, which we call $\Box$-bisimulation.

\begin{definition}[$\Box$-bisimulation]
Let $\M=\lr{S,R,V}$ and $\M'=\lr{S',R',V'}$ be models. A nonempty binary relation $Z$ over $S$ is called a {\em $\Box$-bisimulation} between $\M$ and $\M'$, if $sZs'$ implies the following conditions hold:

(Inv) for all $p\in\BP$, $s\in V(p)$ iff $s'\in V'(p)$;

($\Box$-Forth) if $sRt$, then there exists $t'\in S'$ such that $s'R't'$ and $tZt'$;

($\Box$-Back) if $s'R't'$, then there exists $t\in S$ such that $sRt$ and $tZt'$.

\medskip

We say that $(\M,s)$ and $(\M',s')$ are {\em $\Box$-bisimilar}, notation: $(\M,s)\Kbis(\M',s')$, if there exists a $\Box$-bisimulation $Z$ between $\M$ and $\M'$ such that $sZs'$.
\end{definition}

$\Box$-bisimulation is a notion tailored to standard modal logic $\ML$. Under this notion, it is shown that $\ML$-formulas are invariant under $\Box$-bisimulation, thus $\ML$ cannot distinguish $\Box$-bisimilar models, and on image-finite models (or even $\ML$-saturated models)\footnote{Given a model $\M=\lr{S,R,V}$, $\M$ is said to be {\em image-finite}, if for any $s\in S$, the set $R(s)$ is finite; $\M$ is said to be {\em $\ML$-saturated}, if for any $\Gamma\subseteq \ML$ and any $s\in S$, if all of finite subsets of $\Gamma$ are satisfiable in $R(s)$, then $\Gamma$ is also satisfiable in $R(s)$.}, $\Box$-equivalence coincides with $\Box$-bisimilarity (c.f., e.g. \cite{blackburnetal:2001}). The following property of bisimilarity will be used in the sequel.

\begin{proposition}\label{prop.Kbisimilar} Let $\M=\lr{S,R,V}$ and $\M'=\lr{S',R',V'}$ be two models, and $s\in S$ and $s'\in S'$. Then
$(\M,s)\Kbis(\M',s')$ implies the following conditions:
\begin{enumerate}
\item \label{prop.bisimilar-one} For all $p\in\BP$, $s\in V(p)$ iff $s'\in V'(p)$;
\item \label{prop.bisimilar-two} If $sRt$, then there is a $t'$ in $\M'$ such that $s'R't'$ and $t\Kbis t'$;
\item \label{prop.bisimilar-three} If $s'R't'$, then there is a $t$ in $\M$ such that $sRt$
and $t\Kbis t'$.
\end{enumerate}
\end{proposition}

\begin{proof}
Follows directly from the fact that $\Kbis$ is a $\Box$-bisimulation and the definition of $\Box$-bisimulation.
\end{proof}

However, the notion of $\Box$-bisimulation is too refined for the logic $\SNCL$, as will be shown below. The following  example arises in the proof of Prop.~\ref{prop.lessexp}:

$$
\xymatrix{\mathcal{M}:\ \ \ s:p\ar@(ur,ul) & & & \mathcal{N}:\ \ \ t:p}
$$

It is not hard to show that $\M$ and $\N$ are both image-finite models, and that $(\M,s)$ and $(\N,t)$ satisfy the same $\SNCL$-formulas, but they are not $\Box$-bisimilar. Therefore, we need to redefine a suitable bisimulation notion for $\SNCL$.
\weg{As shown from the example in Proposition~\ref{prop.lessexp}, the image-finite models $(\M,s)$ and $(\N,t)$ satisfy the same $\SNCL$-formulas, but they are not $\Box$-bisimilar. This indicates that the standard notion of bisimulation ($\Box$-bisimulation) is too refined for $\SNCL$. We need to redefine a suitable bisimulation notion for $\SNCL$.}

\begin{definition}[$\blacktriangle$-bisimulation]\label{def.delta-bis}
Let $\M=\langle S,R,V\rangle$. A nonempty binary relation $Z$ over $S$ is called a {\em $\blacktriangle$-bisimulation} on $\M$, if $sZs'$ implies that the following conditions are satisfied:

(Inv) for all $p\in\BP$, $s\in V(p)$ iff $s'\in V(p)$;

($\blacktriangle$-Forth) if $sRt$ and $(s,t)\notin Z$ for some $t$, then there is a $t'$ such that $s'Rt'$ and $tZt'$;

($\blacktriangle$-Back) if $s'Rt'$ and $(s',t')\notin Z$ for some $t'$, then there is a $t$ such that $sRt$ and $tZt'$.

\medskip

We say that $(\M,s)$ and $(\M',s')$ are {\em $\blacktriangle$-bisimilar}, notation: $(\M,s)\kwbis(\M',s')$, if there exists a $\blacktriangle$-bisimulation $Z$ on the disjoint union of $\M$ and $\M'$ such that $sZs'$.
\end{definition}

\begin{proposition}
If $Z$ and $Z'$ are both $\blacktriangle$-bisimulations on $\M$, then $Z\cup Z'$ is also a $\blacktriangle$-bisimulation on $\M$.
\end{proposition}

\begin{proof}
Suppose that $Z$ and $Z'$ are both $\blacktriangle$-bisimulations on $\M$, to show $Z\cup Z'$ is also a $\blacktriangle$-bisimulation on $\M$. Obviously, $Z\cup Z'$ is nonempty, since $Z$ and $Z'$ are both nonempty. We need to check that $Z\cup Z'$ satisfies the three conditions of $\blacktriangle$-bisimulation. For this, assume that $(s,s')\in Z\cup Z'$. Then $sZs'$ or $sZ's'$.

(Inv):  If $sZs'$, then as $Z$ is a $\blacktriangle$-bisimulation, we have: given any $p\in\BP$, $s\in V(p)$ iff $s'\in V(p)$; if $sZ's'$, then as $Z'$ is a $\blacktriangle$-bisimulation, we also have: given any $p\in\BP$, $s\in V(p)$ iff $s'\in V(p)$. In both case we have that given any $p\in\BP$, $s\in V(p)$ iff $s'\in V(p)$.

($\blacktriangle$-Forth): Suppose that $sRt$ and $(s,t)\notin Z\cup Z'$, then $(s,t)\notin Z$ and $(s,t)\notin Z'$. If $sZs'$, then since $Z$ is a $\blacktriangle$-bisimulation on $\M$, there exists $t'$ such that $s'Rt'$ and $(t,t')\in Z$, and hence $(t,t')\in Z\cup Z'$; if $sZ's'$, then since $Z'$ is a $\blacktriangle$-bisimulation on $\M$, there exists $t'$ such that $s'Rt'$ and $(t,t')\in Z'$, and hence also $(t,t')\in Z\cup Z'$. Therefore in both cases, there exists $t'$ such that $s'Rt'$ and $(t,t')\in Z\cup Z'$.

($\blacktriangle$-Back): The proof is similar to that of ($\blacktriangle$-Forth).
\end{proof}

Thus we can build more sophisticated $\blacktriangle$-bisimulations from the simpler $\blacktriangle$-bisimulations. In particular, by Def.~\ref{def.delta-bis}, we can see that $\blacktriangle$-bisimilarity is the largest $\blacktriangle$-bisimulation. And also, $\blacktriangle$-bisimilarity is an equivalence relation. Note that the proof is highly nontrivial.

\begin{proposition}
The $\blacktriangle$-bisimilarity $\kwbis$ is an equivalence relation.
\end{proposition}

\begin{proof}
We need only show that $\kwbis$ satisfies the three properties of an equivalence relation.

Reflexivity: Given any model $\M=\lr{S,R,V}$ and $s\in S$, to show that $(\M,s)\kwbis(\M,s)$. For this, define $Z=\{(w,w)\mid w\in S\}$. First, $Z$ is nonempty, as $sZs$. We need only show that $Z$ satisfies the three conditions of $\blacktriangle$-bisimulation. Suppose that $wZw$.

It is obvious that $w$ and $w$ satisfy the same propositional variables, thus (Inv) holds; suppose that $wRt$ and $(w,t)\notin Z$ for some $t\in S$, then obviously, there exists $t'=t$ such that $wRt'$ and $tZt'$, thus ($\blacktriangle$-Forth) holds; the proof of ($\blacktriangle$-Back) is analogous.

\medskip

Symmetry: Given any models $\M=\lr{S,R,V}$ and $\M'=\lr{S',R',V'}$ and $s\in S$ and $s'\in S'$, assume that $(\M,s)\kwbis (\M',s')$, to show that $(\M',s')\kwbis (\M,s)$. By assumption, we have that there exists $\blacktriangle$-bisimulation $Z$ with $sZs'$. Define $Z'=\{(w,w')\mid w\in S,~w'\in S',~ w'Zw\}\cup\{(w,t)\mid w,t\in S,~wZt\}\cup\{(w',t')\mid w',t'\in S',~w'Zt'\}.$ First, since $sZs'$, we have $(s',s)\in Z'$, thus $Z'$ is nonempty. We need only check that $Z'$ satisfies the three conditions of $\blacktriangle$-bisimulation. Suppose that $wZ'w'$.

By supposition, we have $w'Zw$. Using (Inv) of $Z$, we have that $w'$ and $w$ satisfy the same propositional variables, then of course $w$ and $w'$ satisfy the same propositional variables, thus (Inv) holds. For ($\blacktriangle$-Forth), suppose that $wRt$ and $(w,t)\notin Z'$ for some $t\in S$, then by definition of $Z'$, $(w,t)\notin Z$. Using ($\blacktriangle$-Back) of $Z$, we infer that there exists $t'\in S'$ such that $w'R't'$ and $t'Zt$, thus $tZ't'$. The proof of ($\blacktriangle$-Back) is similar, by using ($\blacktriangle$-Forth) of $Z$.

\medskip

Transitivity: Given any models $\M=\lr{S^{\M},R^{\M},V^{\M}}$,~$\N=\lr{S^{\N},R^{\N},V^{\N}}$,\\
$\mathcal{O}=\lr{S^{\mathcal{O}},R^{\mathcal{O}},V^{\mathcal{O}}}$ and $s\in S^{\M},~t\in S^{\N},~u\in S^{\mathcal{O}}$, assume that $(\M,s)\kwbis(\N,t)$ and $(\N,t)\kwbis(\mathcal{O},u)$, to show that $(\M,s)\kwbis(\mathcal{O},u)$. By assumption, we have that there exists $\blacktriangle$-bisimulation $Z_1$ on the disjoint union of $\M$ and $\N$ such that $sZ_1t$, and there exists $\blacktriangle$-bisimulation $Z_2$ on the disjoint union of $\N$ and $\mathcal{O}$ such that $tZ_2u$. We need to find a $\blacktriangle$-bisimulation $Z$ on the disjoint union of $\M$ and $\mathcal{O}$.

Define $Z=\{(x,z)\mid x\in S^{\M},~z\in S^{\mathcal{O}},\text{ there is a }y\in S^{\N}\text{ such that }xZ_1y,~yZ_2z\}\cup\{(x,x')\mid x,x'\in S^{\M},~xZ_1x'\}\cup\{(z,z')\mid z,z'\in S^{\mathcal{O}},~zZ_2z'\}\cup\{(x,x')\mid x,x'\in S^{\M},\text{ there are }y,y'\in S^{\N}\text{ such that }yZ_2y',~xZ_1y,~x'Z_1y'\}\cup\{(z,z')\mid z,z'\in S^{\mathcal{O}},\text{ there are }y,y'\in S^{\N}\text{ such that }yZ_1y',~yZ_2z,~y'Z_2z'\}.$ First, since $sZ_1t$ and $tZ_2u$, by the first part of the definition of $Z$, we have $sZu$, thus $Z$ is nonempty. We need only check that $Z$ satisfies the three conditions of $\blacktriangle$-bisimulation. Suppose that $xZz$. Then by the first part of the definition of $Z$, there is a $y\in S^{\N}$ such that $xZ_1y$ and $yZ_2z$.

(Inv): as $Z_1$ and $Z_2$ are both $\blacktriangle$-bisimulations, $x$ and $y$ satisfy the same propositional variables, and $y$ and $z$ satisfy the same propositional variables. Then $x$ and $z$ satisfy the same propositional variables.

($\blacktriangle$-Forth): suppose that $xR^{\M}x'$ and $(x,x')\notin Z$ for some $x'\in S^{\M}$, then by the second part of the definition of $Z$, we obtain $(x,x')\notin Z_1$. From this, $xZ_1y$ and ($\blacktriangle$-Forth) of $Z_1$, it follows that there exists $y'\in S^{\N}$ such that $yR^{\N}y'$ and $x'Z_1y'$. Using $xZ_1y,x'Z_1y',(x,x')\notin Z$ and the fourth part of the definition of $Z$, we get $(y,y')\notin Z_2$. From this, $yZ_2z$ and ($\blacktriangle$-Forth) of $Z_2$, it follows that there exists $z'\in S^{\mathcal{O}}$ such that $zR^{\mathcal{O}}z'$ and $y'Z_2z'$. Since $x'Z_1y'$ and $y'Z_2z'$, by the first part of the definition of $Z$, we obtain $x'Zz'$. We have shown that, there exists $z'\in S^{\mathcal{O}} $ such that $zR^{\mathcal{O}}z'$ and $x'Zz'$, as desired.

($\blacktriangle$-Back): the proof is similar to that of ($\blacktriangle$-Forth), but in this case we use the third and fifth parts of the definition of $Z$, rather than the second or fourth parts of the definition of $Z$.
\end{proof}

The following result indicates the relationship between $\blacktriangle$-bisimilarity and $\Box$-bisimilarity: $\blacktriangle$-bisimilarity is strictly weaker than $\Box$-bisimilarity. This corresponds to the fact that $\SNCL$ is strictly weaker than $\ML$.

\begin{proposition}\label{prop.kbisvskwbis}
Let $(\M,s),(\M',s')$ be pointed models. If $(\M,s)\Kbis (\M',s')$, then $(\M,s)\kwbis(\M',s')$; but the converse does not hold.
\end{proposition}

\begin{proof}
Suppose that $(\M,s)\Kbis (\M',s')$. Define $Z=\{(x,x')\mid x\Kbis x'\}$. We will show that $Z$ is a $\blacktriangle$-bisimulation on the disjoint union of $\M$ and $\M'$ with $sZs'$.

First, by supposition, we have $sZs'$, thus $Z$ is nonempty. We need only check that $Z$ satisfies the three conditions of $\blacktriangle$-bisimulation. Assume that $xZx'$. By definition of $Z$, we obtain $x\Kbis x'$. Using item \ref{prop.bisimilar-one} of Proposition \ref{prop.Kbisimilar}, we have $x$ and $x'$ satisfy the same propositional variables, thus (Inv) holds. For ($\blacktriangle$-Forth), suppose that $xRy$ and $(x,y)\notin Z$ for some $y$, then using item \ref{prop.bisimilar-two} of Proposition \ref{prop.Kbisimilar}, we get there exists $y'$ in $\M'$ such that $x'R'y'$ and $y\Kbis y'$, thus $yZy'$. The condition ($\blacktriangle$-Back) is similar to prove, by using item \ref{prop.bisimilar-three} of Proposition \ref{prop.Kbisimilar}.

For the converse, recall the example in Proposition \ref{prop.lessexp}. There, let $Z=\{(s,s),(s,t)\}$. It is not hard to show that $Z$ is a $\blacktriangle$-bisimulation on the disjoint union of $\M$ and $\N$, thus $(\M,s)\kwbis(\N,t)$. However, $(\M,s)\not\Kbis(\N,t)$, as $s\nvDash\Box\bot$ but $t\vDash\Box\bot$.
\end{proof}

The following result says that $\SNCL$-formulas are invariant under $\blacktriangle$-bisimilarity. This means that $\SNCL$-formulas cannot distinguish $\blacktriangle$-bisimilar models.
\begin{proposition}\label{prop.invariance}
For any models $(\M,s)$ and $(\M',s')$,  if $(\M,s)\kwbis(\M',s')$, then $(\M,s)\equkw(\M',s')$. In other words, $\blacktriangle$-bisimilarity implies $\blacktriangle$-equivalence.
\end{proposition}

\begin{proof}
Assume that $(\M,s)\kwbis(\M',s')$, then there is a $\blacktriangle$-bisimulation $Z$ on the disjoint union of $\M$ and $\M'$ such that $sZs'$. We need to show that for any $\phi\in \SNCL$, $\M,s\vDash\phi$ iff $\M',s'\vDash\phi$.

The proof continues by induction on the structure of $\phi$. The non-trivial case is $\blacktriangle\phi$.

Suppose that $\M,s\nvDash\blacktriangle\phi$. Then there exists $t$ such that $sRt$ and $(s\vDash\phi\not\Leftrightarrow t\vDash\phi)$. Without loss of generality, assume that
$s\vDash\phi$ and $t\nvDash\phi$, then $(s,t)\notin Z$ by the induction hypothesis. We obtain by ($\blacktriangle$-Forth) that there exists $t'$ such that $s'R't'$ and $tZt'$, thus $(\M,t)\kwbis (\M',t')$. From $s\kwbis s'$ and the induction hypothesis and $s\vDash\phi$, it follows that $s'\vDash\phi$. Analogously, we can infer $t'\nvDash\phi$. Therefore $\M',s'\nvDash\blacktriangle\phi$. For the converse use ($\blacktriangle$-Back).
\end{proof}

With the notion of $\blacktriangle$-bisimulation, we can simplify the proofs in the previous sections. We here take Proposition \ref{prop.lessexp-d} as an example, to show that $(\M,s)$ and $(\N,s')$ therein are $\blacktriangle$-bisimilar, rather than using simultaneous induction. For this, we define $Z=\{(s,s'),(t,t'),(t,t)\}$\footnote{Note that in order to guarantee $Z$ is indeed a $\blacktriangle$-bisimulation, the pair $(t,t)$ must be contained in $Z$.}. We can show that $Z$ is indeed a $\blacktriangle$-bisimulation on the disjoint union of $\M$ and $\N$, thus $s\kwbis s'$ and $t\kwbis t'$. By Proposition \ref{prop.invariance}, we have for all $\phi\in\SNCL$, (i) $\M,s\vDash\phi$ iff $\N,s'\vDash\phi$, and (ii) $\M,t\vDash\phi$ iff $\N,t'\vDash\phi$.

\medskip

For the converse, we have

\begin{proposition}[Hennessy-Milner Theorem]\label{prop.hennessy-milner-theorem}
For any image-finite models $\M,\M'$ and $s\in\M,s'\in\M'$, $(\M,s)\equkw(\M',s')$ iff $(\M,s)\kwbis(\M',s')$.
\end{proposition}

\begin{proof}
Let $\M$ and $\M'$ be both image-finite models and $s\in\M$ and $s'\in\M'$. Based on Proposition \ref{prop.invariance}, we need only to show the direction from left to right. Assume that $(\M,s)\equkw(\M',s')$, we need to show that $\equkw$ is a $\blacktriangle$-bisimulation on the disjoint union of $\M$ and $\M'$, which implies $(\M,s)\kwbis(\M',s')$. It suffices to show the condition ($\blacktriangle$-Forth), as the proof for ($\blacktriangle$-Back) is similar.

Suppose that there exists $t$ such that $sRt$ and $s\not\equkw t$, to show for some $t'$ it holds that $s'R't'$ and $t\equkw t'$. Since $s\not\equkw t$, there is a $\phi\in\SNCL$ such that $s\vDash\phi$ but $t\nvDash\phi$, and thus $s\nvDash\blacktriangle\phi$ due to $sRt$. By assumption, we have $s'\vDash\phi$ and $s'\nvDash\blacktriangle\phi$, and thus there exists $v'$ such that $s'R'v'$ and $v'\nvDash\phi$. Let $S'=\{t'\mid s'R't'\}$. It is easy to see that $S'\neq \emptyset$. As $\M'$ is image-finite, $S'$ must be finite, say $S'=\{t_1',t_2',\cdots,t_n'\}$. If there is no $t'_i\in S'$ such that $t\equkw t'_i$, then for every $t_i'\in S'$ there exists $\phi_i\in\SNCL$ such that $t\vDash\phi_i$ but $t_i'\nvDash\phi_i$. It follows that $t\vDash\phi_1\land\cdots\land\phi_n$, and thus $t\nvDash\phi_1\land\cdots\land\phi_n\to\phi$; furthermore, from $s\vDash\phi$ follows that $s\vDash\phi_1\land\cdots\land\phi_n\to\phi$. Hence $s\nvDash\blacktriangle(\phi_1\land\cdots\land\phi_n\to\phi)$. Note that for all $t_i'\in S'$, $t_i'\nvDash\phi_1\land\cdots\land\phi_n$, thus $t_i'\vDash\phi_1\land\cdots\land\phi_n\to\phi$. We also have $s'\vDash\phi_1\land\cdots\land\phi_n\to\phi$, and then $s'\vDash\blacktriangle(\phi_1\land\cdots\land\phi_n\to\phi)$, which is contrary to the assumption and $s\nvDash\blacktriangle(\phi_1\land\cdots\land\phi_n\to\phi)$.  Therefore, we have for some $t'$ it holds that $s'R't'$ and $t\equkw t'$.
\end{proof}

If we remove the condition of `image-finite', then $\kwbis$ does not coincide with $\equkw$.

\begin{example}\label{example.not-m-saturated}
Consider two models $\M=\lr{S,R,V}$ and $\M'=\lr{S',R',V'}$, where $S=\mathbb{N}\cup\{s\}$, $R=\{(s,n)\mid n\in\mathbb{N}\},V(p_n)=\{n\}$
and $S'=\mathbb{N}\cup\{s',\omega\}$, $R'=\{(s',n)\mid n\in\mathbb{N}\}\cup\{(s',\omega)\}$, and $V'(p_n)=\{n\}$. This can be visualized as follows:
$$
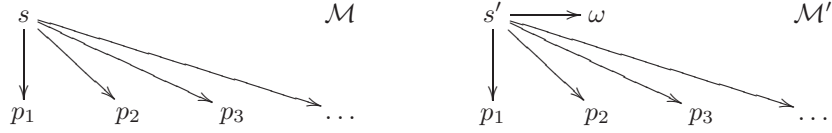
\xymatrix{
s \ar[d]\ar[dr]\ar[drr]\ar[drrr]&&&\M\\
p_1&p_2&p_3& \dots
}
\qquad
\qquad
\xymatrix{
s' \ar[d]\ar[dr]\ar[drr]\ar[drrr]\ar[r]&\omega&&\M'\\
p_1&p_2&p_3& \dots
}
$$
We have:
\begin{itemize}
\item Neither of $\M$ and $\M'$ is image-finite, as $s$ and $s'$ both have infinite many successors.
\item $(\M,s)\equkw(\M',s')$. By induction on $\phi\in\SNCL$, we show: for any $\phi$, $\M,s\vDash\phi$ iff $\M',s'\vDash\phi$. The non-trivial case is $\blacktriangle\phi$, that is to show, $\M,s\vDash\blacktriangle\phi$ iff $\M',s'\vDash\blacktriangle\phi$. The direction from right to left is easy, just noting that $R(s)\subseteq R'(s')$. For the other direction, suppose that $\M,s\vDash\blacktriangle\phi$. Then for any $n\in\mathbb{N}$, we have that $s\vDash\phi$ iff $n\vDash\phi$, which implies that $s'\vDash\phi$ iff $n\vDash\phi$ by the induction hypothesis. As $\phi$ is finite, it contains only finitely many propositional variables. Without loss of generality, we may assume that $n$ is the largest number of subscripts of propositional variables occurring in $\phi$. Then by induction on $\phi$, we can show that $n+1\vDash\phi$ iff $\omega\vDash\phi$. Thus given any $n\in\mathbb{N}\cup\{\omega\}$, $s'\vDash\phi$ iff $n\vDash\phi$. Therefore $\M',s'\vDash\blacktriangle\phi$, as desired.
\item $(\M,s)\not\kwbis(\M',s')$. Suppose towards contradiction that $(\M,s)\kwbis(\M',s')$, then there exists a $\blacktriangle$-bisimulation such that $sZs'$. Now we have $s'R'\omega$. And also $(s',\omega)\notin Z$, for otherwise $s'\kwbis \omega$, thus e.g. $s'\vDash\blacktriangle p_1$ iff $\omega\vDash\blacktriangle p_1$, contrary to the fact that $s'\nvDash\blacktriangle p_1$ but $\omega\vDash\blacktriangle p_1$. By the condition ($\blacktriangle$-Back), we obtain that there exists $m\in\mathbb{N}$ such that $sRm$ and $mZ\omega$, thus $m\kwbis \omega$. However, $m\vDash p_m$ but $\omega\nvDash p_m$, contradiction.
\end{itemize}
\end{example}

We can extend the result of Prop.~\ref{prop.hennessy-milner-theorem} to the following proposition. Here by {\em $\SNCL$-saturated model} we mean, given any $s$ in this model and any set $\Gamma\subseteq\SNCL$, if all of finite subsets of $\Gamma$ are satisfiable in the successors of $s$, then $\Gamma$ is also satisfiable in the successors of $s$.

\begin{proposition}\label{prop.delta-saturated}
Let $(\M,s)$ and $(\M',s')$ be $\SNCL$-saturated pointed models. Then $(\M,s)\equkw(\M',s')$ iff $(\M,s)\kwbis (\M',s')$.
\end{proposition}

\begin{proof}
Based on Proposition \ref{prop.invariance}, we need only show the direction from left to right.

Let $\M=\lr{S,R,V}$ and $\M'=\lr{S',R',V'}$ be $\SNCL$-saturated models. Suppose that $(\M,s)\equkw(\M',s')$, we will show that $\equkw$ is a $\blacktriangle$-bisimulation on the disjoint union of $\M$ and $\M'$, which implies $(\M,s)\kwbis (\M',s')$. It suffices to show the condition ($\blacktriangle$-Forth) holds, as the proof of ($\blacktriangle$-Back) is similar.

Assume that $sRt$ and $s\not\equkw t$ for some $t$, to show there exists $t'$ such that $s'R't'$ and $t\equkw t'$. Let $\Gamma=\{\phi\in\SNCL\mid t\vDash\phi\}$. It is clear that $t\vDash\Gamma$. Then for any finite $\Sigma\subseteq \Gamma$, $t\vDash\bigwedge\Sigma$. As $s\not\equkw t$, there exists $\psi\in\SNCL$ such that $s\vDash\psi$ but $t\nvDash\psi$, thus $s\vDash\bigwedge\Sigma\to\psi$ but $t\nvDash\bigwedge\Sigma\to\psi$, hence $s\nvDash\blacktriangle(\bigwedge\Sigma\to\psi)$. If for any $u'$ with $s'R'u'$ we have $u'\nvDash\bigwedge\Sigma$, then $u'\vDash\bigwedge\Sigma\to\psi$. Since $s\equkw s'$ and $s\vDash\psi$, it follows that $s'\vDash\psi$, thus $s'\vDash\bigwedge\Sigma\to\psi$, hence $s'\vDash\blacktriangle(\bigwedge\Sigma\to\psi)$, contrary to $s\equkw s'$ and $s\nvDash\blacktriangle(\bigwedge\Sigma\to\psi)$. Therefore there exists $u'$ such that $s'R'u'$ and $u'\vDash\bigwedge\Sigma$. Because $\M'$ is $\SNCL$-saturated, for some $t'$ we have $s'R't'$ and $t'\vDash\Gamma$. Furthermore, $t\equkw t'$: given any $\phi\in\SNCL$, if $t\vDash\phi$, then $\phi\in\Gamma$, hence $t'\vDash\phi$; if $t\nvDash\phi$, i.e., $t\vDash\neg\phi$, then $\neg\phi\in\Gamma$, hence $t'\vDash\neg\phi$, i.e., $t'\nvDash\phi$, as desired.
\end{proof}

The condition `$\SNCL$-saturated' is also indispensable, which can also be illustrated with Example \ref{example.not-m-saturated}. In that example, $\M$ is not $\SNCL$-saturated. To see this point, note that the set $\{\neg p_1,\neg p_2,\cdots,\neg p_n\}$ is finitely satisfiable in the successors of $s$, but the set itself is {\em not} satisfiable in the successors of $s$. In the meantime, $(\M,s)\equkw (\M',s')$ but $(\M,s)\not\kwbis(\M',s')$.

\medskip

We have seen from Def.~\ref{def.delta-bis} that the notion of $\blacktriangle$-bisimulation is quite different from that of $\Box$-bisimulation. However, it is surprising that the notion of $\blacktriangle$-bisimulation contraction is very similar to that of $\Box$-bisimulation contraction, by simply replacing $\Kbis$ with $\kwbis$.
\begin{definition}[$\blacktriangle$-bisimulation contraction] Let $\M=\lr{S,R,V}$ be a model. The {\em $\blacktriangle$-bisimulation contraction} of $\M$ is the quotient structure $[\M]=\lr{[S],[R],[V]}$ such that
\begin{itemize}
\item $[S]=\{[s]\mid s\in S\}$, where $[s]=\{t\in S\mid s\kwbis t\}$;
\item $[s][R][t]$ iff there exist $s'\in[s]$ and $t'\in[t]$ such that $s'Rt'$\weg{ and $s'\not\kwbis t'$};
\item $[V](p)=\{[s]\mid s\in V(p)\}$ for all $p\in \BP$.
\end{itemize}
\end{definition}

Under this definition, we obtain that the contracted model (via $\kwbis$) is $\blacktriangle$-bisimilar to the original model, and that the $\mathcal{S}5$-model property is preserved under $\blacktriangle$-bisimulation contraction.

\begin{proposition}
Let $\M=\lr{S,R,V}$ be a model, and let $[\M]=\lr{[S],[R],[V]}$ be the $\blacktriangle$-bisimulation contraction of $\M$. Then for any $s\in S$, we have $([\M],[s])\kwbis(\M,s)$.
\end{proposition}

\begin{proof}
%Define $Z=\{([w],w)\mid w\in S\}\cup\{(w,v)\mid w,v\in S, w\kwbis v\}\cup\{([w],[v])\mid w,v\in S, w\kwbis v\}$.
Define $Z=\{([w],w)\mid w\in S\}\cup\{([w],[v])\mid w,v\in S, w\kwbis v\}$. First, since $S$ is nonempty, $Z$ is nonempty. We need to show that $Z$ satisfies the three conditions of $\blacktriangle$-bisimulation, which entails $([\M],[s])\kwbis(\M,s)$. Assume that $[w]Zw$.

(Inv): by the definition of $[V]$.

($\blacktriangle$-Forth): suppose that $[w][R][v]$ and $([w],[v])\notin Z$, then by definition of $[R]$, there exist $w'\in[w]$ and $v'\in[v]$ such that $w'Rv'$. As $[w]=[w']$ and $[v]=[v']$, we get from the supposition that $([w'],[v'])\notin Z$. By definition of $Z$, we obtain $w'\not\kwbis v'$. Because $\kwbis$ is a $\blacktriangle$-bisimulation, from $w\kwbis w'$ it follows that there exists $u$ such that $wRu$ and $u\kwbis v'$, thus $[u]=[v']=[v]$. It is clear that $[u]Zu$, that is, $[v]Zu$.

($\blacktriangle$-Back): suppose that $wRv$ and $(w,v)\notin Z$. By definition of $[R]$, we have $[w][R][v]$. Obviously, $[v]Zv$.
\end{proof}

\begin{proposition}
Let $\M=\lr{S,R,V}$ be a model, and let $[\M]=\lr{[S],[R],[V]}$ be the $\blacktriangle$-bisimulation contraction of $\M$. If $\M$ is an $\mathcal{S}5$-model, then $[\M]$ is also an $\mathcal{S}5$-model.
\end{proposition}

\begin{proof}
Suppose that $\M=\lr{S,R,V}$ is an $\mathcal{S}5$-model, to show that $[\M]=\lr{[S],[R],[V]}$ is also an $\mathcal{S}5$-model. We need to show that $[R]$ is an equivalence relation, that is, $[R]$ satisfies the properties of reflexivity, symmetry and transitivity.
The nontrivial case is transitivity. For this, given any $[s],[t],[u]\in[S]$, assume that $[s][R][t]$ and $[t][R][u]$, we need only show that $[s][R][u]$.

By assumption and the definition of $[R]$, there exists $s'\in[s],t'\in[t]$ such that $s'Rt'$, and there exists $t''\in[t],u'\in[u]$ such that $t''Ru'$. We now consider two cases:
\begin{itemize}
\item $t''\kwbis u'$. In this case, we have $[t'']=[u]$, then $[t]=[u]$, thus $[s][R][u]$.
\item $t''\not\kwbis u'$. In this case, using $t'\kwbis t''$ (as $t'\in[t]$ and $t''\in[t]$) and the fact that $\kwbis$ is a $\blacktriangle$-bisimulation, we obtain that there exists $u''$ such that $t'Ru''$ and $u''\kwbis u'$, thus $u''\in[u']=[u]$. Moreover, From $s'Rt',t'Ru''$ and the transitivity of $R$, it follows that $s'Ru''$. We have thus shown that there exists $s'\in[s],u''\in[u]$ with $s'Ru''$, therefore $[s][R][u]$.
\end{itemize}
In both cases we have $[s][R][u]$, as desired.
\end{proof}

\subsection{Characterization Results}

As $\vDash\blacktriangle\phi\leftrightarrow(\phi\to\Box\phi)\land(\neg\phi\to\Box\neg\phi)$, strong noncontingency logic $\SNCL$ can be seen as a fragment of standard modal logic $\ML$, and also a fragment of first-order logic. In this subsection we characterize strong noncontingency logic within standard modal logic and within first-order logic. To make our exposition self-contained, we introduce some definitions and results from e.g. \cite{blackburnetal:2001} without proofs.

\begin{definition}[Ultrafilter Extension]\label{def.ue} Let $\M=\lr{S,R,V}$ be a model. We say that $ue(\M)=\lr{Uf(S),R^{ue},V^{ue}}$ is the \emph{ultrafilter extension} of $\M$, if
\begin{itemize}
\item $Uf(S)=\{u\mid u \text{ is an ultrafilter over }S\}$, where an ultrafilter $u\subseteq \mathcal{P}(S)$ satisfies the following properties:
      \begin{itemize}
      \item $S\in u$, $\emptyset\notin u$,
      \item $X,Y\in u$ implies $X\cap Y\in u$,
      \item $X\in u$ and $X\subseteq Z\subseteq S$ implies $Z\in u$
      \item For all $X\in\mathcal{P}(S)$, $X\in u$ iff $S\backslash X\notin u$% ($-X$ means the complement of $X$)
      \end{itemize}
\item For all $s,t\in Uf(S)$, $sR^{ue}t$ iff for all $X\subseteq S$, $X\in t$ implies $\lambda(X)\in s$, where $\lambda(X)=\{w\in S\mid\text{ there exists }v \text{ such that }wRv \text{ and }v\in X\}$
\item $V^{ue}(p)=\{u\in Uf(S)\mid V(p)\in u\}$.
\end{itemize}
\end{definition}

\begin{definition}[Principle Ultrafilter] Let $S$ be a non-empty set. Given any $s\in S$, the \emph{principle ultrafilter} $\pi_s$ generated by $s$ is defined by $\pi_s=\{X\subseteq S\mid s\in X\}$. It can be shown that every principle ultrafilter is an ultrafilter.
\end{definition}

\begin{proposition}\label{prop.equk}
Let $\M$ be a model and $ue(\M)$ be its ultrafilter extension. Then $ue(\M)$ is m-saturated and $(\M,s)\equk (ue(\M),\pi_s)$.
\end{proposition}

Since $\SNCL$ can be viewed as a fragment of $\ML$, every $\SNCL$-formula can be seen as an $\ML$-formula. Thus we have

\begin{lemma}\label{lem.equkw}
Let $\M$ be a model and $s\in\M$. Then $ue(\M)$ is $\SNCL$-saturated and $(\M,s)\equkw(ue(\M),\pi_s)$.
\end{lemma}

From Lemma \ref{lem.equkw} and Proposition \ref{prop.delta-saturated}, it follows that
\begin{lemma}\label{lem.ue}
Let $(\M,s)$ and $(\N,t)$ be pointed models. Then $(\M,s)\equkw(\N,t)$ implies $(ue(\M),\pi_s)\kwbis(ue(\N),\pi_t)$.
\end{lemma}

We are now close to prove two characterization results: strong noncontingency logic is the $\blacktriangle$-bisimulation-invariant fragment of standard modal logic and of first-order logic. In the following, by an $\ML$-formula $\phi$ (resp. a first-order formula $\alpha$) {\em is invariant under $\blacktriangle$-bisimulation}, we mean for any models $(\M,s)$ and $(\N,t)$, if $(\M,s)\kwbis(\N,t)$, then $\M,s\vDash\phi$ iff $\N,t\vDash\phi$ (resp. $\M,s\vDash \alpha$ iff $\N,t\vDash\alpha$).
\begin{theorem}\label{thm.delta-characterization}
An $\ML$-formula is equivalent to an $\SNCL$-formula iff it is invariant under $\blacktriangle$-bisimulation.
\end{theorem}

\begin{proof}
Based on Prop.~\ref{prop.invariance}, we need only show the direction from right to left. For this, suppose that an $\ML$-formula $\phi$ is invariant under $\blacktriangle$-bisimulation.

Let $MOC(\phi)=\{t(\psi)\mid \psi\in\SNCL,\phi\vDash t(\psi)\}$, where $t$ is a translation function which recursively translates every $\SNCL$-formulas into the corresponding $\ML$-formulas. In particular, $t(\blacktriangle\psi)=(t(\psi)\to\Box t(\psi))\land(t(\neg\psi)\to\Box \neg t(\psi))$.

If we can show that $MOC(\phi)\vDash\phi$, then by Compactness Theorem of modal logic, there exists a finite set $\Gamma\subseteq MOC(\phi)$ such that $\bigwedge\Gamma\vDash \phi$, i.e., $\vDash\bigwedge\Gamma\to \phi$. Besides, the definition of $MOC(\phi)$ implies that $\phi\vDash\bigwedge\Gamma$, i.e., $\vDash \phi\to\bigwedge\Gamma$, and thus $\vDash\bigwedge\Gamma\lra\phi$. Since every $\gamma\in\Gamma$ is a translation of an $\SNCL$-formula, so is $\Gamma$. Then we are done.

Assume that $\M,s\vDash MOC(\phi)$, to show that $\M,s\vDash\phi$. Let $\Sigma=\{t(\psi)\mid \psi\in\SNCL,\M,s\vDash t(\psi)\}$. We now claim $\Sigma\cup\{\phi\}$ is satisfiable: otherwise, by Compactness Theorem of modal logic again, there exists finite $\Sigma'\subseteq \Sigma$ such that $\phi\vDash\neg\bigwedge\Sigma'$, thus $\neg\bigwedge\Sigma'\in MOC(\phi)$. By assumption, we obtain $\M,s\vDash\neg\bigwedge\Sigma'$. However, the definition of $\Sigma$ and $\Sigma'\subseteq \Sigma$ implies $\M,s\vDash\bigwedge\Sigma'$, contradiction.

Thus we may assume that $\N,t\vDash\Sigma\cup\{\phi\}$. We can show $(\M,s)\equkw(\N,t)$ as follows: for any $\psi\in\SNCL$, if $\M,s\vDash\psi$, then $\M,s\vDash t(\psi)$, and then $t(\psi)\in\Sigma$, thus $\N,t\vDash t(\psi)$, hence $\N,t\vDash\psi$; if $\M,s\nvDash\psi$, i.e., $\M,s\vDash\neg\psi$, then $\M,s\vDash t(\neg\psi)$, and then $t(\neg\psi)\in\Sigma$, thus $\N,t\vDash t(\neg\psi)$, hence $\N,t\vDash\neg\psi$, i.e. $\N,t\nvDash\psi$.

We now construct the ultrafilter extensions of $\M$ and $\N$, denoted by $ue(\M)$ and $ue(\N)$, respectively. According to the fact that $(\M,s)\equkw(\N,t)$ and Lemma \ref{lem.ue}, we have $(ue(\M),\pi_s)\kwbis(ue(\N),\pi_t)$. Since $\N,t\vDash\phi$, by Lemma \ref{lem.equkw}, we have $ue(\N),\pi_t\vDash\phi$. From supposition it follows that $ue(\M),\pi_s\vDash\phi$. Using Lemma \ref{lem.equkw} again, we conclude that $\M,s\vDash\phi$.
\end{proof}

\begin{theorem}
A first-order formula is equivalent to an $\SNCL$-formula iff it is invariant under $\blacktriangle$-bisimulation.
\end{theorem}

\begin{proof}
Based on Prop.~\ref{prop.invariance}, we need only show the direction from right to left. For this, suppose that a first-order formula $\alpha$ is invariant under $\blacktriangle$-bisimulation, then by Prop.~\ref{prop.kbisvskwbis}, we have that $\alpha$ is also invariant under $\Box$-bisimulation. From van Benthem Characterization Theorem (cf. e.g. \cite{blackburnetal:2001}), it follows that $\alpha$ is equivalent to an $\ML$-formula $\phi$. From this and supposition, it follows that $\phi$ is invariant under $\blacktriangle$-bisimulation. Applying Thm.~\ref{thm.delta-characterization}, $\phi$ is equivalent to an $\SNCL$-formula. Therefore, $\alpha$ is equivalent to an $\SNCL$-formula.
\end{proof}

\section{Axiomatization: Minimal system}\label{sec.minimal}

In this section, we present a complete axiomatic system for $\SNCL$ over the class of all frames.

\subsection{Axiomatic system and soundness}

\begin{definition} [Axiomatic system \SLCL] The axiomatic system \SLCL\ consists of all propositional tautologies (\TAUT), uniform substitution (\SUB), modus ponens (\MP), plus the following axioms and inference rule:
\[
\begin{array}{ll}
\KwTop& \blacktriangle\top\\

\EquiKw&\blacktriangle\neg p\leftrightarrow\blacktriangle p\\

\KwCon&\blacktriangle p\land\blacktriangle q\to\blacktriangle (p\land q)\\

\R&\text{From }\phi\to\psi \text{ infer }\blacktriangle\phi\land\phi\to\blacktriangle\psi\\
\end{array} \]

A {\em derivation} from $\Gamma$ to $\phi$ in $\SLCL$, notation: $\Gamma\vdash_{\SLCL}\phi$, is a finite sequence of $\SNCL$-formulas in which each formula is either an instantiation of an axiom, or an element of $\Gamma$, or the result of applying an inference rule to prior formulas in the sequence. Formula $\phi$ is {\em provable} in $\SLCL$, or a {\em theorem}, notation: $\vdash\phi$, if there is a derivation from the empty set $\emptyset$ to $\phi$ in $\SLCL$.
\end{definition}

Intuitively, Axiom \KwTop\ means that tautologies are always strongly noncontingent; Axiom \EquiKw\ says that a formula is strongly noncontingent is the same as its negation is strongly noncontingent; Axiom \KwCon\ states that if two formulas are strongly noncontingent, then so is their conjunction; Rule \R\ stipulates the \emph{almost} monotonicity of the strong noncontingency operator.

Note that when it comes to completeness, Axiom \KwTop\ is indispensable in the system \SLCL, because the subsystem $\SLCL-\KwTop$ is incomplete. To see this, define an auxiliary semantics $\Vdash$ as the same as $\vDash$, except that all formulas of the form $\blacktriangle\phi$ are interpreted as \emph{false}. One can check the system $\SLCL-\KwTop$ is sound with respect to the new semantics $\Vdash$, but $\KwTop$ is \emph{not} valid, which entails that \KwTop\ is not provable in $\SLCL-\KwTop$. On the other hand, \KwTop\ is valid under the standard semantics $\vDash$. Therefore, $\SLCL-\KwTop$ is not complete with respect to the semantics $\vDash$.

\medskip

By induction on $n\in\mathbb{N}$, using Axiom \KwCon\ we have
\begin{fact}\label{fact.two}
For all $n\in\mathbb{N}$, $\vdash\blacktriangle\phi_1\land\cdots\land\blacktriangle\phi_n\to\blacktriangle(\phi_1\land\cdots\land\phi_n)$.
\end{fact}

\weg{\begin{proposition}
\REKw: $\dfrac{\phi\leftrightarrow\psi}{\blacktriangle\phi\leftrightarrow\blacktriangle\psi}$ is admissible in \SLCL.
\end{proposition}

\begin{proof}
Suppose that $\vdash\phi\leftrightarrow\psi$. Then $\vdash\phi\to\psi$. By \R, we have $\vdash\blacktriangle\phi\land\phi\to\blacktriangle\psi$. From the supposition, it follows that $\vdash\neg\phi\to\neg\psi$. By \R\ again, we obtain $\vdash\blacktriangle\neg\phi\land\neg\phi\to\blacktriangle\neg\psi$.
\end{proof}

\begin{proposition}
\KwT: $\blacktriangle \phi\land\blacktriangle(\phi\to \psi)\land \phi\to\blacktriangle \psi$ is provable in \SLCL.
\end{proposition}

\begin{proof} The following is a sequence of proofs in \SLCL:
\[
\begin{array}{ll}

\end{array}
\]
\end{proof}

\weg{\begin{proposition}
The following rule is admissible in \SLCL, where $n\in\mathbb{N}$: $$\dfrac{(\psi_1\land\cdots\land\psi_n)\to\phi}{(\blacktriangle\psi_1\land\psi_1\land\cdots\blacktriangle\psi_n\land\psi_n)\to\blacktriangle\phi}$$
\end{proposition}

\begin{proof}
By induction on $n$.
\begin{itemize}
\item $n=0$. We need to show that: $\vdash\phi$ implies $\vdash\blacktriangle\phi$.

Assume that $\vdash\phi$, then $\vdash\top\to\phi$. By \R, we have $\vdash\blacktriangle\top\land\top\to\blacktriangle\phi$. From \TAUT\ and \KwTop, it follows that $\vdash\blacktriangle\phi$.
\end{itemize}
\end{proof}}}

\begin{proposition}[Soundness]\label{prop.soundness}
The axiomatic system \SLCL\ is sound with respect to the class of all frames.
\end{proposition}

\begin{proof}
The validity of Axioms \KwTop\ and \EquiKw\ are immediate from the semantics.

For Axiom \KwCon, given an arbitrary model $\mathcal{M}=\langle S,R,V\rangle$ and any $s\in S$, suppose that $\mathcal{M},s\vDash\blacktriangle p\land\blacktriangle q$, to show $\M,s\vDash\blacktriangle(p\land q)$. First assume that $s\vDash p\land q$. Let $t\in S$ be arbitrary such that $sRt$. By assumption $s\vDash p$ and supposition $s\vDash\blacktriangle p$, we obtain $t\vDash p$. By assumption $s\vDash q$ and supposition $s\vDash\blacktriangle q$, we have $t\vDash q$, and then $t\vDash p\land q$. We have thus shown that $s\vDash p\land q$ implies for all $t$ with $sRt$, $t\vDash p\land q$. Analogously, we can show that $s\nvDash p\land q$ implies for all $t$ with $sRt$, $t\nvDash p\land q$. Therefore, $s\vDash\blacktriangle(p\land q)$.
For Rule \R, assume $\vDash\phi\to\psi$, we need to show $\vDash\blacktriangle\phi\land\phi\to\blacktriangle\psi$. Given any model $\mathcal{M}=\langle S,R,V\rangle$ and any $s\in S$, suppose that $\mathcal{M},s\vDash\blacktriangle\phi\land\phi$, to show $\mathcal{M},s\vDash\blacktriangle\psi$. By supposition $\mathcal{M},s\vDash\phi$ and assumption, we obtain $\mathcal{M},s\vDash\psi$. For all $t$ such that $sRt$, from supposition follows that $t\vDash \phi$, and then $t\vDash\psi$ in virtue of the assumption. Moreover, it is clear that $s\nvDash\psi$ implies for all $t$ with $sRt$, $t\nvDash\psi$. Therefore, $\M,s\vDash\blacktriangle\psi$.
\end{proof}

\subsection{Canonical model and completeness}

We proceed with the completeness proof for $\SLCL$. The completeness is shown by a construction of canonical model.
\begin{definition}[Canonical Model]\label{def.cm} $\mathcal{M}^c=\langle S^c, R^c, V^c\rangle$ is the \emph{canonical model }for system $\SLCL$, if
\begin{itemize}
\item $S^c=\{s\mid s\text{ is a maximal consistent set of }\SLCL\}$.
\item For any $s,t\in S^c$, $sR^ct$ iff for all $\phi$, if $\phi\in s$, then $\blacktriangle\phi\in s$ implies $\phi\in t$.
\item $V^c(p)=\{s\in S^c\mid p\in s\}$.
\end{itemize}
\end{definition}

The definition of canonical relation $R^c$ is inspired by the almost-equivalence schema AE (see Prop.~\ref{prop.ad}). Recall that in the canonical model for standard modal logic $\ML$, the canonical relation is always defined by $sR^ct$ iff for all $\phi$, $\Box\phi\in s$ implies $\phi\in s$. Now according to the schema AE, we can replace $\Box\phi\in s$ with $\blacktriangle\phi\in s$ given $\phi\in s$.

\begin{lemma}[Truth Lemma]\label{lem.truthlem}
For any $\phi\in\SNCL$ and any $s\in S^c$, we have $$\mathcal{M}^c,s\vDash\phi\Longleftrightarrow\phi\in s.$$
\end{lemma}

\begin{proof}
By induction on $\phi$. The only non-trivial case is $\blacktriangle\phi$.

$\Longleftarrow:$ Assume towards contradiction that $\blacktriangle\phi\in s$ but $\mathcal{M},s\nvDash\blacktriangle\phi$. From the semantics and the inductive hypothesis, it follows that ($\phi\in s$ and there exists $t\in S^c$ such that $sR^ct$ and $\phi\notin t$) or ($\phi\notin s$ and there exists $t'\in S^c$ such that $sR^ct'$ and $\phi\in t'$). If the first case holds, then from $\phi\in s$ and the assumption $\blacktriangle\phi\in s$ and $sR^ct$, it follows that $\phi\in t$, contradiction. If the second case holds, then since $\phi\notin s$, we have $\neg\phi\in s$. By $\blacktriangle\phi\in s$ again, Axiom $\EquiKw$ and Rule $\SUB$, we obtain $\blacktriangle\neg\phi\in s$, thus $\neg\phi\land\blacktriangle\neg\phi\in s$, from which and $sR^ct'$ we deduce $\neg\phi\in t'$, contradiction.

$\Longrightarrow:$ Suppose that $\blacktriangle\phi\notin s$, we need to show that $\mathcal{M}^c,s\nvDash \blacktriangle\phi$, by the inductive hypothesis, that is to show, ($\phi\in s$ and there exists $t\in S^c$ such that $sR^ct$ and $\neg\phi\in t$) or ($\neg\phi\in s$ and there exists $t'\in S^c$ such that $sR^ct'$ and $\phi\in t'$). First, we show
\begin{enumerate}
\item\label{a} $\{\psi\mid\blacktriangle\psi\land\psi\in s\}\cup\{\neg\phi\}$ is consistent, and
\item\label{b} $\{\psi\mid\blacktriangle\psi\land\psi\in s\}\cup\{\phi\}$ is consistent.
\end{enumerate}

For \ref{a}, if the set is inconsistent, then there exist $\psi_1,\cdots,\psi_n\in\{\psi\mid\blacktriangle\psi\land\psi\in s\}$ such that $\vdash\psi_1\land\cdots\land\psi_n\to\phi$.\footnote{Note that Axiom \KwTop\ provides the non-emptiness of the set $\{\psi\mid\blacktriangle\psi\land\psi\in s\}$.} By Rule \R, we have $\vdash\blacktriangle(\psi_1\land\cdots\psi_n)\land(\psi_1\land\cdots\psi_n)\to\blacktriangle\phi$. Since $\blacktriangle\psi_i\land\psi_i\in s$ for all $i\in[1,n]$, from Fact~\ref{fact.two}, it follows that $\blacktriangle(\psi_1\land\cdots\land\psi_n)\land(\psi_1\land\cdots\land\psi_n)\in s$, thus we have $\blacktriangle\phi\in s$, contrary to the supposition.

The proof for \ref{b} is similar, in this case we need to use Axiom \EquiKw\ and Rule \SUB.

Therefore we complete the proofs of \ref{a} and \ref{b}. From \ref{a}, the definition of $R^c$ and the observation that every consistent set can be extended to a maximal consistent set (Lindenbaum's Lemma), we obtain that there exists $t\in S^c$ such that $sR^ct$ and $\neg\phi\in t$. Similarly, from \ref{b}, we get that there exists $t'\in S^c$ such that $sR^ct'$ and $\phi\in t'$. Since either $\phi\in s$ or $\neg\phi\in s$, we conclude that $\mathcal{M}^c,s\nvDash \blacktriangle\phi$ based on the previous analysis.
\end{proof}

\begin{proposition}\label{prop.comp} Let $X$ be a normal modal logic between $K$ and $T$: $K\subseteq X\subseteq T$.
The following conditions are equivalent: given any $\Gamma\cup\{\phi\}\subseteq\SNCL$,

$(1)~~\Gamma\vdash_{\SLCL}\phi$ (there is a derivation from $\Gamma$ to $\phi$ in $\SLCL$)

$(2)~~\Gamma\vDash\phi$ ($\Gamma$ entails $\phi$ over the class of all frames)

$(3)~~\Gamma\vDash_X\phi$ ($\Gamma$ entails $\phi$ over the class of frames for $X$)

$(4)~~\Gamma\vDash_T\phi$ ($\Gamma$ entails $\phi$ over the class of reflexive frames)
\end{proposition}

\begin{proof}
$(1)\Rightarrow (2)$: by Prop.~\ref{prop.soundness}.

$(2)\Rightarrow (3)$: obviously.

$(3)\Rightarrow (4)$: since $X\subseteq T$, we have $\vDash_X\subseteq~\vDash_T$.

$(4)\Rightarrow (1)$: suppose that $\Gamma\nvdash_{\SLCL}\phi$, then $\Gamma\cup\{\neg\phi\}$ is $\SLCL$-consistent. By Lindenbaum's Lemma, there is an $s\in S^c$ such that $\Gamma\cup\{\neg\phi\}\subseteq s$, i.e. $\Gamma\subseteq s$ and  $\phi\notin s$. By Lemma \ref{lem.truthlem}, $\M^c,s\vDash\Gamma$ but $\M^c,s\nvDash\phi$. Moreover, that $R^c$ is reflexive is immediate from the definition of $R^c$. Therefore, $\Gamma\nvDash_T\phi$.
\end{proof}

The following are now immediate from Proposition \ref{prop.comp}.

\begin{theorem}\label{thm.k}
\SLCL\ is sound and strongly complete with respect to the class of all frames.
\end{theorem}

\begin{theorem}
\SLCL\ is sound and strongly complete with respect to the class of serial frames.
\end{theorem}

\begin{theorem}
\SLCL\ is sound and strongly complete with respect to the class of reflexive frames.
\end{theorem}

We have seen from Proposition \ref{prop.lessexp} that there is a truth-preserving translation function from $\SNCL$ to $\ML$. Thus $\SNCL$ can be seen as a fragment of $\ML$. Since $\ML$ is decidable, we have also the decidability of $\SNCL$.

\begin{proposition}
The logic of strong noncontingency $\SNCL$ is decidable.
\end{proposition}

\section{Axiomatization: extensions}\label{sec.extension}

In this section, we give extensions of $\SLCL$ over various frame classes and prove their completeness.
Definition \ref{def.ext} indicates extra axioms and the corresponding systems, with in the last column of the table the frame classes for which we will demonstrate completeness.

\begin{definition}[Extensions of \SLCL]\label{def.ext}
\[
\begin{array}{|l|l|l|l|}
\hline
\text{Notation}&\text{Axioms}&\text{Systems}&\text{Frame classes}\\
\hline
\KwTr& \blacktriangle p\to\blacktriangle\blacktriangle p& \SLCLTr=\SLCL+\KwTr&4~(\mathcal{S}4)\\
\KwB& p\to\blacktriangle(\blacktriangle p\to p)&\SLCLB=\SLCL+\KwB&\mathcal{B} ~(\mathcal{TB})\\
\KwEuc&\neg\blacktriangle p\to\blacktriangle\neg\blacktriangle p&\SLCLBEuc=\SLCLB+\KwEuc&\mathcal{B}5~(\mathcal{S}5)\\
\KwEucp&p\land\neg\blacktriangle p\to\blacktriangle(p\land\blacktriangle p)&\SLCLBEucp=\SLCLB+\KwEucp&\mathcal{B}5~(\mathcal{S}5)\\
\hline
\end{array}
\]
\end{definition}

Analogous to the case of $\SLCL$, we can show that Axiom $\KwTop$ is indispensable in the system $\SLCLTr$. However, Axiom $\KwTop$ is dispensable in $\SLCLB$, and thus in $\SLCLBEuc$, since it is provable in $\SLCLB-\KwTop$, just letting $p$ in Axiom $\KwB$ be $\top$.\footnote{Note that we need first show the rule $\dfrac{\phi\leftrightarrow\psi}{\blacktriangle\phi\leftrightarrow\blacktriangle\psi}$ is admissible in $\SLCLB-\KwTop$. For the proof refer to Proposition \ref{prop.re-admissible}, where we do not use Axiom $\KwTop$.} Thus we can replace $\SLCLB$ and $\SLCLBEuc$ by $\SLCLB-\KwTop$ and $\SLCLBEuc-\KwTop$, respectively. We here still write $\SLCLB$ and $\SLCLBEuc$ to keep consistency.

To show the soundness of systems $\SLCLTr$ and $\SLCLB$ with respect to the corresponding classes of frames, by Proposition \ref{prop.soundness}, it suffices to show
\begin{proposition}\label{prop.validities}\
\begin{itemize}
\item \KwTr\ is valid on the class of transitive frames.
\item \KwB\ is valid on the class of symmetric frames.
%\item \KwEuc\ is valid on the class of Euclidean frames.
\end{itemize}
\end{proposition}

\begin{proof} The validity of $\KwB$ on symmetric frames is shown in Prop.~\ref{prop.definable-b}.% We only need to show the case for $\KwTr$.

For \KwTr: given any transitive model $\mathcal{M}=\langle S,R,V\rangle$ and any $s\in S$, assume towards contradiction that $\mathcal{M},s\vDash\blacktriangle p$ but $\mathcal{M},s\nvDash\blacktriangle\blacktriangle p$. By assumption, it must be the case that there exists $t\in S$ such that $sRt$ and $t\vDash\neg\blacktriangle p$. If $s\vDash p$, then from $s\vDash\blacktriangle p$ and $sRt$ we have $t\vDash p$, and thus there exists $u$ such that $tRu$ and $u\nvDash p$. By the transitivity of $R$, we have $sRu$, which is contrary to the fact that $s\vDash\blacktriangle p\land p$ and $u\nvDash p$. If $s\nvDash p$, with the similar argument we can also derive a contradiction.
\end{proof}

\begin{theorem}[Completeness of $\SLCLTr$]\label{thm.tr}
\SLCLTr\ is sound and strongly complete with respect to the class of $4$-frames, thus it is also sound and strongly complete with respect to the class of $\mathcal{S}4$-frames.
\end{theorem}

\begin{proof}
Soundness is clear from the soundness of \SLCL\ and Prop.~\ref{prop.validities}.

For completeness, by Thm.~\ref{thm.k}, we only need to show that $R^c$ is transitive (recall that $R^c$ is reflexive).

Given any $s,t,u\in S^c$, suppose that $sR^ct$ and $tR^cu$, we need to show $sR^cu$. Let $\phi$ be arbitrary, and assume that $\blacktriangle\phi\land\phi\in s$. From the assumption and $sR^ct$ it follows that $\phi\in t$. From $\blacktriangle\phi\in s$, Axiom \KwTr\ and Rule \SUB, we get $\blacktriangle\blacktriangle\phi\in s$. We have shown that $\blacktriangle\blacktriangle\phi\land\blacktriangle\phi\in s$. Then by $sR^ct$ again, we deduce $\blacktriangle\phi\in t$, thus $\phi\land\blacktriangle\phi\in t$. From this and $tR^cu$ it follows that $\phi\in u$. Therefore $sR^cu$.
\end{proof}

\begin{theorem}[Completeness of $\SLCLB$]\label{thm.completeness-b}
\SLCLB\ is sound and strongly complete with respect to the class of $\mathcal{B}$-frames, thus it is also sound and strongly complete with respect to the class of $\mathcal{TB}$-frames.
\end{theorem}

\begin{proof}
Soundness is immediate from the soundness of \SLCL\ and the validity of $\KwB$ (Prop.~\ref{prop.validities}).

For the completeness, by Thm.~\ref{thm.k}, we only need to show that $R^c$ is symmetric (recall that $R^c$ is reflexive).

Given any $s,t\in S^c$. Suppose that $sR^ct$, to show $tR^cs$.  For this, given any $\phi$, and assume that $\phi\notin s$, the only thing is to show that $\blacktriangle\phi\land\phi\notin t$. By assumption, we have $\neg\phi\in s$, and thus $\blacktriangle\neg\phi\to\neg\phi\in s$. From Axiom \KwB\ and Rule \SUB, it follows that $\blacktriangle(\blacktriangle\neg\phi\to\neg\phi)\in s$. By supposition, we obtain $\blacktriangle\neg\phi\to\neg\phi\in t$. By Axiom $\EquiKw$ and Rule $
\SUB$, we get $\blacktriangle\phi\to\neg\phi\in t$, and then we conclude that $\blacktriangle\phi\land\phi\notin t$, as desired.
\end{proof}

%\KwEuc: $\neg\blacktriangle p\to\blacktriangle\neg\blacktriangle p$

One may expect that the same story goes with Axiom \KwEuc\ and system $\SLCL+\KwEuc$. Although \KwEuc\ can provide the completeness of $\SLCL+\KwEuc$ with respect to the class of (reflexive and) Euclidean frames (see the proof of Thm.~\ref{thm.comp-kb5}), unfortunately, it cannot provide the soundness of the system in question, because \KwEuc\ is {\em not} valid on the class of Euclidean frames (or $\mathcal{S}5$-frames).
\begin{proposition}\label{prop.invalidity}
\KwEuc\ is not valid on the class of $5$-frames (or $\mathcal{S}5$-frames).
\end{proposition}

\begin{proof}
Consider the following Euclidean (and reflexive) model $\mathcal{M}$:

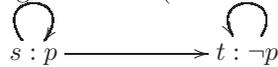
$$
\xymatrix{s:p\ar[rr]\ar@(ul,ur)&& t:\neg p\ar@(ur,ul)
}$$
\weg{First, it is easy to see that $\mathcal{M},s\vDash p\land\Diamond \neg p$. By Fact~\ref{fact.one}, we obtain $s\vDash\blacktriangledown p$, i.e. $s\vDash\neg\blacktriangle p$.}First, as $\mathcal{M},s\vDash p$, $\M,t\nvDash p$ and $sRt$, we have $s\vDash\neg\blacktriangle p$.
Second, since $t$ can only `see' itself, we get $t\vDash\blacktriangle p$, i.e. $t\nvDash\neg\blacktriangle p$, thus $s\nvDash\blacktriangle\neg\blacktriangle p$.\weg{ hence $s\vDash\Diamond\blacktriangle p$. Then $s\vDash\neg\blacktriangle p\land\Diamond\blacktriangle p$. From this and Fact~\ref{fact.one}, it follows that $s\vDash\blacktriangledown\neg\blacktriangle p$, i.e. $s\nvDash\blacktriangle\neg\blacktriangle p$.} Therefore $\mathcal{M},s\nvDash\neg\blacktriangle p\to\blacktriangle\neg\blacktriangle p$.
\end{proof}

We have seen that $\KwEuc$ is not valid on the class of Euclidean frames (or $\mathcal{S}5$-frames). However, this formula is indeed valid on the class of symmetric and Euclidean frames, as shown below.

\begin{proposition}\label{prop.kweuc-valid}
$\KwEuc$ is valid on the class of $\mathcal{B}5$-frames.
\end{proposition}

\begin{proof}
Suppose not, that is, there exists symmetric and Euclidean model $\M=\lr{S,R,V}$ and $s\in S$ such that $\M,s\vDash\neg\blacktriangle p$ but $\M,s\nvDash\blacktriangle\neg\blacktriangle p$. From supposition, it follows that for some $t$ such that $sRt$ and $\M,t\vDash\blacktriangle p$. By symmetry of $R$,  $tRs$. We consider two cases. If $t\vDash p$, then using the fact that $t\vDash\blacktriangle p$ and $tRs$, we obtain that $s\vDash p$. From this and $s\vDash\neg\blacktriangle p$, then there exists $u\in S$ such that $sRu$ and $\M,u\nvDash p$. From $sRt$ and $sRu$ and  the Euclidicity of $R$, it follows that $tRu$, then using $t\vDash\blacktriangle p$ and $t\vDash p$, we obtain $u\vDash p$, contradiction. If $t\nvDash p$, we can also get a contradiction by a similar argument.
\end{proof}

The axiomatization for $\SNCL$ over Euclidean frames is quite involved. Recall that unlike Axiom $\KwB$, Axiom $\KwEuc$ cannot define the corresponding frame property, because the property of Euclidicity is undefinable in $\SNCL$; unlike Axiom $\KwTr$, Axiom $\KwEuc$ is not valid on the corresponding frames. Despite this, we indeed have the following completeness result.

\begin{theorem}[Completeness of $\SLCLBEuc$]\label{thm.comp-kb5}
$\SLCLBEuc$ is sound and strongly complete with respect to the class of $\mathcal{B}5$-frames, thus it is also sound and strongly complete with respect to the class of $\mathcal{S}5$-frames.
\end{theorem}

\begin{proof}
Soundness is immediate from the soundness of $\SLCLB$ and Prop.~\ref{prop.kweuc-valid}. We show the completeness.

By Thm.~\ref{thm.completeness-b}, we only need to show that $R^c$ is Euclidean. %(recall that $R^c$ is reflexive).

Suppose for any $s,t,u\in S^c$ that $sR^ct,sR^cu$, we need show that $tR^cu$. For this, assume for any $\phi$ that $\blacktriangle\phi\land\phi\in t$, we have to show that $\phi\in u$. From $sR^ct$ and $\blacktriangle\phi\in t$, it follows that $\neg(\neg\blacktriangle\phi\land\blacktriangle\neg\blacktriangle\phi)\in s$, then $\blacktriangle\phi\in s$, namely $\blacktriangle\neg\phi\in s$ (for otherwise by Axiom $\KwEuc$ and Rule $\SUB$ we get $\neg\blacktriangle\phi\land\blacktriangle\neg\blacktriangle\phi\in s$, contradiction). Moreover, from $sR^ct$ and $\phi\in t$, we obtain $\neg(\neg\phi\land\blacktriangle\neg\phi)\in s$, and thus $\neg\phi\notin s$, that is, $\phi\in s$. We have thus shown that $\blacktriangle\phi\land\phi\in s$. Therefore $\phi\in u$ from $sR^cu$.
\end{proof}

\weg{Since $R^c$ is reflexive, it is immediate that
\begin{corollary}
$\SLCLBEuc$ is sound and strongly complete with respect to the class of $\mathcal{S}5$-frames.
\end{corollary}}

\weg{\begin{theorem}\label{thm.euc}
\SLCLEuc\ is sound and strongly complete with respect to the class of Euclidean frames.
\end{theorem}}

Consider the formula $\KwEucp:~ p\land\neg\blacktriangle p\to\blacktriangle(p\land\blacktriangle p)$.

\begin{proposition}\label{prop.valid-5}
$\KwEucp$ is valid on the class of Euclidean frames.
\end{proposition}

\begin{proof}
Given any Euclidean model $\M=\lr{S,R,V}$ and any $s\in S$, suppose towards contradiction that $\M,s\vDash p\land\neg\blacktriangle p$ but $\M,s\nvDash\blacktriangle(p\land\blacktriangle p)$. From the first supposition, there exists $t$ such that $sRt$ and $t\nvDash p$. Because $s\nvDash p\land\blacktriangle p$, by the second supposition, there exists $u$ such that $sRu$ and $u\vDash p\land\blacktriangle p$. Using $sRu,sRt$ and the Euclidicity of $R$, we obtain $uRt$, thus $t\vDash p$, contradiction.
\end{proof}

Thus  $\SLCL+\KwEucp$ is sound with respect to the class of Euclidean frames (and also $\mathcal{S}5$-frames). We are unsure whether $\SLCL+\KwEucp$ is also {\em strongly complete} with respect to the class of Euclidean frames. We leave this for future work.

\medskip

Nevertheless, $\SLCLBEucp$ is sound and strongly complete with respect to the class of symmetric and Euclidean frames.
\begin{theorem}[Completeness of $\SLCLBEucp$]
$\SLCLBEucp$ is sound and strongly complete with respect to the class of $\mathcal{B}5$-frames, thus it is also sound and strongly complete with respect to the class of $\mathcal{S}5$-frames.
\end{theorem}

\begin{proof}
By Thm.~\ref{thm.completeness-b} and Prop.~\ref{prop.valid-5}, we only need to show that $R^c$ is Euclidean.

Suppose for any $s,t,u\in S^c$ that $sR^ct$ and $sR^cu$, to show $tR^cu$. For this, assume for any $\phi$ that $\phi\land\blacktriangle\phi\in t$, we need only show $\phi\in u$. From the supposition $sR^ct$ and the symmetry of $R^c$, it follows that $tR^cs$. From this and the assumption, we have $\phi\in s$.

We can show that $\blacktriangle\phi\in s$: otherwise, $\blacktriangle\phi\notin s$, then $\neg\blacktriangle\phi\in s$. Since $\phi\in s$, by using Axiom $\KwEucp$ and Rule $\SUB$, we get $\blacktriangle(\phi\land\blacktriangle\phi)\in s$, i.e. $\blacktriangle\neg(\phi\land\blacktriangle\phi)\in s$; moreover, $\neg(\phi\land\blacktriangle\phi)\in s$, then from $sR^ct$ follows that $\neg(\phi\land\blacktriangle\phi)\in t$, contrary to the assumption.

We have thus shown that $\blacktriangle\phi\in s$. Thus $\phi\land\blacktriangle\phi\in s$, and hence $\phi\in u$ due to $sR^cu$.
\end{proof}

\section{The related work}\label{sec.relatedwork}
In \cite[page 119]{DBLP:journals/ndjfl/Humberstone02}, Humberstone proposed an alternative semantics for the agreement operator $O$. Given a model $\M=\langle S,R,V\rangle$ and a world $s\in S$,
$$\M,s\vDash O\phi \Leftrightarrow \forall t\in S(sRt\Rightarrow (\M,s\vDash\phi\Leftrightarrow\M,t\vDash\phi)).$$
%Intuitively, the truth value of $\phi$ is agreed on at $s$, if and only if, $s$ and its all successors agree on $\phi$.
Intuitively, $O\phi$ is true at $s$, if and only if, $s$ and its all successors agree on the truth value of $\phi$. One can show that the semantics of $O$ and $\blacktriangle$ are equivalent.

Humberstone then gave the following axiomatization. Although motivated in a different setting, \cite[page 119]{DBLP:journals/ndjfl/Humberstone02} also axiomatized $\SNCL$ over the class of arbitrary frames. We call it $\mathbf{LA}$, wherein we have replaced $O$ by $\blacktriangle$.\footnote{In~\cite[note 24]{DBLP:journals/ndjfl/Humberstone02}, it is pointed out that ${\mathbf A3}$ can also be replaced with $\blacktriangle\phi\to\blacktriangle(\phi\vee\psi)\lor(\neg\phi\land\blacktriangle(\neg\phi\vee\chi))$ or $\blacktriangle\phi\to(\phi\land \blacktriangle(\phi\vee\psi))\lor\blacktriangle(\neg\phi\vee\chi)$.}
\[\begin{array}{ll}
\mathbf{PL}& \text{all substitution instances of tautologies}\\
{\mathbf A1}& \blacktriangle\phi\to\blacktriangle\neg\phi\\
{\mathbf A2} & \blacktriangle\phi\land\blacktriangledown(\phi\lor\psi)\to\blacktriangledown\psi\\
{\mathbf A3} & \blacktriangle\phi\to(\phi\land \blacktriangle(\phi\vee\psi))\lor(\neg\phi\land\blacktriangle(\neg\phi\vee\chi))\\
{\mathbf R\blacktriangle}& \text{From } \phi \text{ infer } \blacktriangle\phi\\
\mathbf{MP} & \text{From } \phi \text{ and } \phi\to\psi \text{ infer } \psi\\
\mathbf{RE} & \text{From } \phi\leftrightarrow\psi \text{ infer } \blacktriangle\phi\leftrightarrow\blacktriangle\psi\\
\end{array}\]

Inspired by Kuhn's definition of $\lambda$ in the completeness proof for minimal noncontingency logic \cite[page 232]{DBLP:journals/ndjfl/Kuhn95}, Humberstone defined a particular function $\lambda'$ such that, for each maximal consistent set $s$, $\lambda'(s)=\{\phi\mid \phi\land O(\phi\vee\psi)\in s\text{ for all }\psi\}$, and then defined canonical relation $R^c$ as $sR^ct$ iff $\lambda'(s)\subseteq t$. Then he claimed that the system $\mathbf{LA}$ is sound and complete with respect to the class of all Kripke frames.
\weg{By an adaption of Kuhn's argument in the setting of noncontingency logic, Humberstone claimed that the system $\mathbf{LA}$ is sound and complete with respect to the class of all Kripke frames. }%(Actually, ${\mathbf A3}$ can be replaced with ${\mathbf A3'}:~\blacktriangle\phi\to\blacktriangle(\phi\vee\psi)\lor\blacktriangle(\neg\phi\vee\chi)$ (?).) No! Because A3 is not provable from A3', just thinking of the standard semantics of non-contingency  operator.

We tend to find our axiomatization \SLCL\ simpler, e.g. with respect to ${\mathbf A3}$ above. Below we will show that the axioms/rules of $\mathbf{LA}$ are all provable/admissible in $\SLCL$. In the following we write $\vdash$ for $\vdash_{\SLCL}$.

\begin{proposition}
${\mathbf A3}$ is provable in $\SLCL$.
%$\vdash_{\SLCL}\blacktriangle\phi\to(\phi\land \blacktriangle(\phi\vee\psi))\lor(\neg\phi\land\blacktriangle(\neg\phi\vee\chi)).$
\end{proposition}

\begin{proof}
\[
\begin{array}{lll}
(i)& \phi\to\phi\vee\psi & \TAUT\\
(ii)& \blacktriangle\phi\land\phi\to\blacktriangle(\phi\vee\psi) & (i),\R\\
(iii) & \blacktriangle\phi\land\phi\to\phi\land\blacktriangle(\phi\vee\psi)& (ii),\TAUT\\
(iv) &\neg\phi\to\neg\phi\vee\chi & \TAUT\\
(v) & \blacktriangle\neg\phi\land\neg\phi\to\neg\phi\land\blacktriangle(\neg\phi\vee\chi) &(iv),\text{ similar to }(i)-(iii)\\
(vi) & \blacktriangle\phi\to(\phi\land \blacktriangle(\phi\vee\psi))\lor(\neg\phi\land\blacktriangle(\neg\phi\vee\chi))& (iii),(v),\EquiKw,\SUB,\TAUT\\
\end{array}
\]
\end{proof}

\begin{proposition}
${\mathbf R\blacktriangle}$ is admissible in $\SLCL$.
\end{proposition}

\begin{proof}
Suppose that $\vdash\phi$. Then by Axiom $\TAUT$, we have $\vdash\top\to\phi$. Applying Rule $\R$, we get $\vdash\blacktriangle\top\land\top\to\blacktriangle\phi$. Using Axiom $\KwTop$ and Rule $\MP$, we obtain $\vdash\blacktriangle\phi$.
\end{proof}
\weg{\begin{proof}
Suppose that $\vdash\phi$. Then by Axiom $\TAUT$, we have $\vdash\phi\leftrightarrow\top$. Using $\mathbf{RE}$ (Prop.~\ref{prop.re-admissible}), we get $\vdash\blacktriangle\phi\leftrightarrow\blacktriangle\top$, and hence we obtain $\vdash\blacktriangle\phi$ due to Axiom $\KwTop$.
\end{proof}}

\begin{proposition}\label{prop.re-admissible}
$\mathbf{RE}$ is admissible in \SLCL.
\end{proposition}

\begin{proof}
Suppose that $\vdash\phi\leftrightarrow\psi$, we need to show that $\vdash\blacktriangle\phi\leftrightarrow\blacktriangle\psi$. From supposition and Axiom $\TAUT$, it follows that $\vdash\phi\to\psi$. With this and Rule $\R$, we have $\vdash\blacktriangle\phi\land\phi\to\blacktriangle\psi$. Similarly, using $\vdash\neg\phi\to\neg\psi$, we can obtain that $\vdash\blacktriangle\neg\phi\land\neg\phi\to\blacktriangle\neg\psi$. Then from Axioms $\EquiKw,\TAUT$ and Rule $\SUB$, we infer that $\vdash\blacktriangle\phi\to\blacktriangle\psi$. From supposition we also have $\vdash\psi\to\phi$. By the similar argument, we can get $\vdash\blacktriangle\psi\to\blacktriangle\phi$, thus we conclude that $\vdash\blacktriangle\phi\leftrightarrow\blacktriangle\psi$.
\end{proof}

In this paper, we gave a simply minimal axiomatization for $\SNCL$, and its extensions with respect to other various classes of frames, and we advanced the research by comparing the relative expressivity of $\SNCL$ and $\ML$, and that of $\SNCL$ and $\NCL$. Moreover, we proposed a notion of bisimulation for $\SNCL$, based on which we characterized its expressive power within modal logic and within first-order logic. We also have results about the frame (un)definability.

\weg{\section{Discussion}

This work is related to Hintikka semantics for questions \cite{Hintikka:1976}, where for example, the sentence ``John remember whether it is raining'' is equivalent to  ``If it is raining then John remembers that it is raining, and if it is not raining then John remembers that it is not raining.'' This formal analysis is criticized in \cite{Karttunen:1977}, but again adopted by researchers, see e.g., \cite{wangetal:2010}.

\cite[page 6]{Karttunen:1977} lists nine types of question embedding verbs, such as verbs of retaining knowledge (e.g. remember), verbs of acquiring knowledge (e.g. learn), verbs of communication (e.g. tell), inquisitive verbs (e.g. wonder), etc.. Remember whether $\phi$ is interpreted as ``if $\phi$ is true, then remember that $\phi$ is true; if $\phi$ is false, then remember that $\phi$ is false''; however, if `wonder whether' is interpreted in the same way, then it will become if $\phi$ is true, then wonder that $\phi$ is true; if $\phi$ is false, then wonder that $\phi$ is false'', although `wonder' takes only interrogative complements \cite{Egre:2008}. For Hamblin, (direct) questions denotes the set of propositions expressed by their possible answers \cite{Hamblin:1973}. In Karttunen's opinion, ``questions denote the set of true propositions expressed by their true answers'' \cite[page 10]{Karttunen:1977}.}

\section{Conclusion and future work}\label{sec.conclu}

In this paper, motivated by Hintikka's treatment of question embedding verbs, and also the variations of noncontingency operator, we proposed a logic of strong noncontingency $\SNCL$, with strong noncontingency operator as a sole primitive modality. This logic is not normal, since it does not have the validity to the effect that modality is closed under material implication. We compared the relative expressive powers of $\SNCL$, modal logic, and noncontingency logic. We demonstrated that in $\SNCL$, some of the five basic frame properties are undefinable  but some is indeed definable. We proposed suitable notions of bisimulation and of bisimulation contraction for $\SNCL$. Based on this bisimulation, we characterized $\SNCL$ as the invariant fragment of standard modal logic and of first-order logic. Last but not least, we presented Hilbert-style axiomatizations for $\SNCL$ over various frame classes, where the minimal system is simpler than that proposed in \cite{DBLP:journals/ndjfl/Humberstone02}.

For further research we leave the question whether $\SLCL+\KwEucp$ is strongly complete with respect to the class of Euclidean frames. As we have observed, the soundness of $\SLCL+\KwEucp$ follows from the soundness of $\SLCL$ and Proposition \ref{prop.valid-5}. Another direction is characterize $\CLtwo$ within $\SNCL$, since $\CLtwo$ is less expressive than $\SNCL$ on non-reflexive models (see Proposition \ref{prop.exp.clncl}). A work related to this is to compare the notion of bisimulations of $\SNCL$ and of $\CLtwo$. From the proof of Proposition \ref{prop.exp.clncl}, we can check that $(\M,s)\Deltabis(\N,t)$ but $(\M,s)\not\kwbis(\N,t)$, thus $\Deltabis$ is not stronger than $\kwbis$. But how about the other way around? And also, one can consider logics in which the combinations of some modalities involved here are contained, such as the combination of $\blacktriangle$ and $\Delta$, and that of $\circ$ and $\Delta$.

Last but not least, we can compare the relative expressivity of $\SNCL$ and $\LEA$. In Section \ref{sec.exp}, we compare the relative expressive powers of $\SNCL$ and $\ML$, and that of $\SNCL$ and $\CLtwo$, but we did not compare the relative expressivity of $\SNCL$ and $\LEA$. For this, we observe that,
since $\vDash\blacktriangle\phi\leftrightarrow\circ\phi\land\circ\neg\phi$, there is a truth-preserving translation from $\SNCL$ to $\LEA$, thus $\LEA$ is at least as expressive as $\SNCL$. \weg{And we can show that $\SNCL$ and $\LEA$ are equally expressive on the class of $\mathcal{T}$-models, since $\blacktriangle\phi\leftrightarrow\circ\phi$ is valid on $\mathcal{T}$-models. (not valid!)}But we do not know whether $\LEA$ is {\em more} expressive than $\SNCL$. Although one can show that $\LEA$ shares the same notion of bisimulation as $\SNCL$ (e.g., Definition \ref{def.delta-bis} and Propositions \ref{prop.invariance},~\ref{prop.hennessy-milner-theorem} also apply to $\LEA$, see \cite{Fan:2015:essence}),  which means the two languages have the same distinguishing power, this does not mean they are equally expressive,\weg{ Bisimulation is a notion about distinguishing power.} because even two languages that have the same distinguishing power does not necessarily have the same expressive power. For example, the full language of propositional logic has the same distinguishing power as its fragment which has only propositional variables, but they are not equally expressive. We conjecture that $\LEA$ is more expressive than $\SNCL$, since it seems to us that even the simplest $\LEA$-formula $\circ p$ cannot be expressed with $\SNCL$. We leave this for future work.

\weg{Another work to do is compare the notion of bisimulations of $\SNCL$ and of $\CLtwo$. From the proof of Proposition \ref{prop.lessexp}, we can check that $(\M,s)\Deltabis(\N,t)$ but $(\M,s)\not\kwbis(\N,t)$, thus $\Deltabis$ is not stronger than $\kwbis$. But how about the other way around? In other words, is $\kwbis$ stronger than $\Deltabis$? We do not know the answer.}

\bibliographystyle{plain}
\bibliography{biblio2014}

\end{document}